\documentclass[journal,twocolumn,final]{IEEEtran}
\usepackage{caption}
\usepackage{amsthm,epsfig,url,latexsym,mathrsfs,amsfonts,array,graphicx,cite,amssymb,setspace,subfig}
\usepackage{algorithm,algorithmic,balance}
\usepackage{psfrag}
\usepackage[cmex10]{amsmath}
\DeclareMathOperator*{\nom}{nom}
\DeclareMathOperator*{\rob}{rob}
\DeclareMathOperator*{\block}{block}

\DeclareMathOperator*{\mat}{mat}
\DeclareMathOperator*{\VEC}{vec}

\DeclareMathOperator*{\subto}{s.t.}

\DeclareMathOperator*{\SINr}{SINR_1}
\DeclareMathOperator*{\SINR}{SINR_2}
\DeclareMathOperator*{\WFq}{RWF_q}
\DeclareMathOperator*{\WFQ}{RWF}

\DeclareMathOperator*{\Diag}{Diag}
\DeclareMathOperator*{\PoA}{PoA}
\DeclareMathOperator*{\der}{\partial}
\DeclareMathOperator*{\D}{\frac{\partial}{\partial\epsilon}}
\DeclareMathOperator*{\E}{E}
\DeclareMathOperator*{\crit}{o}

\def\nn{\nonumber}
\newtheorem{thrm}{Theorem}
\newtheorem{lemma}{Lemma}
\newtheorem{corr}{Corollary}

\theoremstyle{remark}
\newtheorem{Remark}{Remark}
\newtheorem{prop}{\hspace{1cm}Proposition}
\floatname{algorithm}{\hspace{0.25cm}\textbf{Algorithm}}
\title{Robust Rate-Maximization Game Under\\ Bounded Channel Uncertainty}
\author{Amod J.G. Anandkumar, ˜\IEEEmembership{Student~Member,~IEEE}, Animashree Anandkumar, ˜\IEEEmembership{Member,~IEEE},\\ Sangarapillai Lambotharan, ˜\IEEEmembership{Senior~Member,~IEEE}, and Jonathon A. Chambers, ˜\IEEEmembership{Fellow,~IEEE}. \thanks{
A.J.G. Anandkumar, S. Lambotharan and J.A. Chambers are with the Advanced Signal Processing Group, School of Electronic, Electrical and Systems Engineering, Loughborough University, Loughborough, LE11 3TU, U.K. (Email: \{A.J.G.Anandkumar, S.Lambotharan, J.A.Chambers\}@lboro.ac.uk). 

A. Anandkumar is with the Center for Pervasive Communications and Computing, Department of Electrical Engineering and Computer Science, Henry Samueli School of Engineering, University of California Irvine, Irvine, CA  92697-2625 USA (Email: a.anandkumar@uci.edu).

This work was supported by EPSRC grant EP/F065477/1. The second author was supported in part by ARO Grant W911NF-06-1-0076. Part of this work was presented at the 2010 IEEE International Conference on Acoustics, Speech, and Signal Processing \cite{conf:Amod_Icassp10} and the Fourty-Third Asilomar Conference on Signals, Systems and Computers \cite{conf:Amod_Asilomar10}.}
}
\begin{document}
\maketitle
\begin{abstract}
We consider the problem of decentralized power allocation for competitive rate-maximization in a frequency-selective Gaussian interference channel under bounded channel uncertainty. We formulate a distribution-free robust framework for the rate-maximization game. We present the robust-optimization equilibrium for this game and derive sufficient conditions for its existence and uniqueness.  We show that an iterative waterfilling algorithm converges to this equilibrium under certain sufficient conditions. We analyse the social properties of the equilibrium under varying channel uncertainty bounds for the two-user case. We also observe an interesting phenomenon that the equilibrium moves towards a frequency-division multiple access solution for any set of channel coefficients under increasing channel uncertainty bounds. We further prove that increasing channel uncertainty can lead to a more efficient equilibrium, and hence, a better sum rate in certain two-user communication systems. Finally, we confirm, through simulations, this improvement in equilibrium efficiency is also observed in systems with a higher number of users.
\end{abstract}
\section{Introduction}\label{sec:intro}
Rate-maximization is an important signal processing problem for power-constrained multi-user wireless systems. It involves solving a power control problem for mutually interfering users operating across multiple frequencies in a Gaussian interference channel. In modern wireless systems where users may enter or leave the system freely and make decisions independently, decentralized control approaches and distributed algorithms are necessary. Game-theoretic methods provide an appropriate set of tools for the design of such algorithms and have been increasingly used for the analysis and study of communications problems \cite{bk:mackenzie2006game}.

In multi-user systems, the users can either cooperate with each other to achieve a socially optimal solution or compete against one another to optimize their own selfish objectives. Cooperative game-theoretic approaches to the problem of power control in wireless networks have been investigated in \cite{jrnl:Zehavi_CoopGT}, \cite{jrnl:Berry_JSAC06} and surveyed in \cite{mag:GT_FF_GIC, mag:GT_FS_GIC, mag:Dist_RA_games}.  In this paper, we are interested in the other scenario where the users are competing against one another and aiming to maximize their own information rates. This competitive rate-maximization problem can be modelled as a strategic noncooperative game. The Nash equilibrium \cite{bk:osborne99a} of this game can be achieved via a distributed waterfilling algorithm where each user performs waterfilling by considering the multi-user interference as an additive coloured noise.  However, most of the current results on rate-maximization and waterfilling algorithms assume the availability of perfect information which is a strong requirement and cannot be met by practical systems.

This paper addresses the following fundamental questions: i) How can the users independently allocate power if the channel state information (CSI) they have is imperfect? How can we formulate a rate maximization game under channel uncertainty and what is the nature of the equilibrium of this game? ii) What are the existence and uniqueness properties of such an equilibrium? How can such a solution be computed by a distributed algorithm and what are the conditions for asymptotic convergence of such an algorithm? iii) How are these conditions affected by the channel uncertainty? iv) What is the effect of uncertainty on the sum-rate and the price of anarchy of such a system? In answering these questions, we can gain further insight into the behaviour of waterfilling algorithms and methods to improve sum-rate in general.
\subsection{Summary of Main Results}
The main contributions of this paper are three-fold. First, we provide a game-theoretic solution for the problem of competitive rate-maximization in the presence of channel uncertainty. Secondly, we analyse the efficiency of the equilibrium as a function of the channel uncertainty bound and prove that the robust waterfilling solution proposed achieves a higher sum-rate with increasing uncertainty under certain conditions. Finally, we verify these results via simulations. We show that improved sum-rate under channel uncertainty entails a cost in terms of more stringent conditions for the uniqueness of the equilibrium and slower convergence of the iterative algorithm to the equilibrium.

For systems with bounded channel uncertainty, we present a distribution-free robust formulation of the rate-maximization game based on an ellipsoid approximation of the uncertainty in the CSI. We present the robust-optimization equilibrium (RE) for this game. At the equilibrium, the users perform a modified waterfilling operation where frequency-overlap among users is penalized. We derive sufficient conditions for the existence and uniqueness of the equilibrium and for the convergence of an asynchronous iterative waterfilling algorithm to the equilibrium of this game.

In our work, we investigate the effect of channel uncertainty on the sum-rate of the system for the two-user case under two scenarios, viz., a two-frequency system and a system with large number of frequencies. For the two-frequency system, the equilibrium sum-rate improves and the price of anarchy decreases as the channel uncertainty increases under high interference. On the other hand, the behavior is reversed at low interference. Another important contribution of our work is to show that increasing channel uncertainty always drives the equilibrium closer to a frequency division multiple access (FDMA) solution for any set of channel coefficients for the system with asymptotic number of frequencies. This is because the users become more conservative about causing interference under increased uncertainty and this leads to  better partitioning of the frequencies among the users. Under certain channel conditions, this also translates to an improvement in the sum-rate and a decrease in the price of anarchy of the system. Thus, we show an interesting phenomenon where increased channel uncertainty can lead to a more efficient equilibrium, and hence, a better sum-rate in certain multi-user communication systems.
\subsection{Related Work}
An iterative waterfilling algorithm for maximizing information rates in digital subscriber line systems \cite{jrnl:cioffi} is an early application of a game-theoretic approach to designing a decentralized algorithm for multi-user dynamic power control. This framework has been further analyzed and extended in \cite{jrnl:luo2006analysis, jrnl:shum2007convergence, jrnl:Etkin2007, jrnl:AIWFA, jrnl:KTH}. The inefficiency of the Nash equilibrium (which need not be Pareto optimal) has been addressed in \cite{jrnl:Blum_TSP03} and \cite{jrnl:Opt1} and methods to improve the sum-rate of the system by using various pricing schemes and modified utility functions have been presented in \cite{jrnl:PriceIWFA, jrnl:KTH, conf:KTH_Gcom09}.  A centralized controller maximizing the sum-rate of the system leads to a non-convex optimization problem and has been shown to be strongly NP--hard in \cite{jrnl:ZQLuo_ComplexityDuality}. The Pareto-optimality of the FDMA solution for this sum-rate maximization problem under certain channel conditions has been proved in \cite{jrnl:ZQLuo}.

Uncertainty in game theory and distributed optimization problems has only recently been investigated. The issue of bounded uncertainty in specific distributed optimization problems in communication networks has been investigated in \cite{conf:Infocom08_DRO} wherein techniques to define the uncertainty set such that they can be solved distributively by robust optimization techniques are presented. In \cite{jrnl:robGameTh}, incomplete-information finite games have been modelled as a distribution-free \emph{robust game} where the players use a robust optimization approach to counter bounded payoff uncertainty. This robust game model also introduced a distribution-free equilibrium concept called the \emph{robust-optimization equilibrium} on which our approach is based.

A brief look at a robust optimization approach for the rate-maximization game with uncertainty in the noise-plus-interference estimate has been presented in \cite{jrnl:Haykin_RobIWFA}, where the authors present a numerically computed algorithm unlike the closed form results presented here. Such a numerical solution prevents further mathematical analysis of the equilibrium and its behaviour under different uncertainty bounds. Also, this uncertainty model is different from ours, where we assume the availability of CSI of the interfering channels and that these quantities have a bounded uncertainty.

A similar problem of rate-maximization in the presence of uncertainty in the estimate of noise-plus-interference levels due to quantization in the feedback channel has been considered in \cite{jrnl:RIWFA}. This problem has been solved using a probabilistically constrained optimization approach and as in our work, also results in the waterfilling solution moving closer to an FDMA solution, with corresponding improvement in sum-rate. However, the effect of quantization on the conditions for existence and uniqueness of the equilibrium and convergence of the algorithm have not been considered. The results presented in \cite{jrnl:RIWFA} are for a sequentially updated algorithm whereas our results allow asynchronous (and thus sequential or simultaneous) updates to the algorithm. Also, the power allocations computed by such a probabilistic optimization formulation do not guarantee that the information rates expected will be achieved for all channel realizations, unlike our worst-case optimization formulation. Furthermore, the relative error (and not just the absolute error due to quantization) in the interference estimate as defined in \cite{jrnl:RIWFA} is assumed to be bounded and drawn from a uniform distribution, which is inaccurate. 
In addition, this bound on the relative error can only be computed if the noise variance at the receivers is assumed to be known (which is not the case). The bounds computed in such a fashion are very loose and will degrade system performance. The other assumption that this relative error bound is in the range $[0,1)$ means that the absolute quantization error has to be less than the noise variance at the receivers, which restricts the applicability of the approach. Our problem formulation has no such limitation on the uncertainty bound based on the noise variances in the system.

Robust rate-maximization for a cognitive radio scenario with uncertainty in the channel to the primary user has been presented in \cite{conf:Palomar_RobCR}. This leads to a noncooperative game formulation without any uncertainties in the payoff functions of the game (unlike in our case) with robust interference limits acting as a constraint on the admissible set of strategies. This game is then solved by numerical optimization as there is no closed form solution.

Initial results of our work, involving the problem formulation of the robust rate-maximization game and the results for sufficient conditions for existence and uniqueness of the equilibrium, along with the conditions for asymptotic convergence of the waterfilling algorithm to the equilibrium, have been presented in \cite{conf:Amod_Icassp10}. Initial results on the analysis of the efficiency of the equilibrium was presented in \cite{conf:Amod_Asilomar10}. We now extend this work to include detailed proofs and further analysis of the efficiency of the equilibrium as a function of the channel uncertainty bound.
\subsection{Paper Outline}\label{subsec:outline}
This paper is organized as follows: Section~\ref{sec:sysModel} describes the system model and provides the necessary preliminaries. Section~\ref{sec:RobGameForm} formulates the robust game model for the rate-maximization game for the single-input single-output (SISO) frequency-selective Gaussian interference channel. Section~\ref{sec:RoE} presents the robust-optimization equilibrium for this game and conditions for its existence and uniqueness, along with the conditions for asymptotic convergence of the waterfilling algorithm to the equilibrium. Section~\ref{sec:Efficiency} presents the analysis on the effect of uncertainty on the sum-rate of the system for the two-user scenario. Section~\ref{sec:results} presents simulation results and Section~\ref{sec:concl} presents the conclusions from this work and possible future research directions.
\section{System Model and Preliminaries}\label{sec:sysModel}
\subsubsection*{Notations used} Vectors and matrices are denoted by lowercase and uppercase boldface letters respectively. The operators $(\cdot)^T, (\cdot)^{-T}, \E \{ \cdot \}$ and $\|\cdot\|_2 $ are  respectively transpose, transpose of matrix inverse, statistical expectation and Euclidean norm operators. The diagonal matrix with the arguments as diagonal elements is denoted by $\Diag(\cdot)$. The quantity $[\mathbf{A}]_{ij}$ refers to the $(i,j)$-th element of $\mathbf{A}$. $\mathbb{R}_+^{m \times n}$ is the set of $m \times n$ matrices with real non-negative elements. The spectral radius (largest absolute eigenvalue) of matrix $\mathbf{A}$ is denoted by $\rho(\mathbf{A})$. The operation $[x]_a^b$ is defined as $[x]_a^b = a \textrm{ if } x \leq a; \ x \textrm{ if } a < x < b; \ b \textrm{ if } x \geq b$ and $(x)^+ \triangleq \max(0,x)$. The matrix projection onto the convex set $\mathscr{Q}$ is denoted by $[\mathbf{x}]_{\mathscr{Q}} \triangleq \arg \min_{\mathbf{z} \in \mathscr{Q}} \|\mathbf{z}-\mathbf{x}\|_2$. The term $\mathbf{w}>\mathbf{0}$ indicates that all elements of $\mathbf{w}$ are positive and $\mathbf{X} \succ \mathbf{0}$ indicates that the matrix $\mathbf{X}$ is positive definite.

We consider a system similar to the one in \cite{jrnl:AIWFA}, which is a SISO frequency-selective Gaussian interference channel with $N$ frequencies, composed of $Q$ SISO links. $\Omega \triangleq \{1,\dots,Q \}$ is the set of the $Q$ players (i.e. SISO links). The quantity $H_{rq}(k)$ denotes the frequency response of the $k$-th frequency bin of the channel between source $r$ and destination $q$.\label{def:Hrq}
The variance of the zero-mean circularly symmetric complex Gaussian noise at receiver $q$ in the frequency bin $k$ is denoted by $\bar{\sigma}_q^2(k)$. The channel is assumed to be quasi-stationary for the duration of the transmission. Let $\sigma_q^2(k) \triangleq \bar{\sigma}_q^2(k)/|H_{qq}(k)|^2$ and the total transmit power of user $q$ be $P_q$. Let the vector $\mathbf{s}_q \triangleq [s_q(1)\ s_q(2) \dots s_q(N)]$ be the $N$ symbols transmitted by user $q$ on the $N$ frequency bins and $p_q(k) \triangleq \E \{ |s_q(k) |^2 \}$ be the power allocated to the $k$-th frequency bin by user $q$ and $\mathbf{p}_q \triangleq [p_q(1)\ p_q(2) \dots p_q(N)]$ be the power allocation vector. The power allocation of each user $q$ has two constraints:
\begin{itemize}\label{pow_constrain}
  \item Maximum total transmit power for each user
  \begin{equation}\label{eq:S_cons1}
    \E \big\{ \|\mathbf{s}_q \|_2^2 \big\} = \sum_{k=1}^{N} p_q(k) \leq P_q,
  \end{equation}
  for $q=1, \dots ,Q,$ where $P_q$ is power in units of energy per transmitted symbol.
  \item Spectral mask constraints
  \begin{equation}\label{eq:S_cons2}
    \E \big\{ |s_q(k) |^2 \big\} = p_q(k) \leq p_q^{max}(k),
  \end{equation}
  for $k=1, \dots, N$ and $q=1, \dots, Q$, where $p_q^{max}(k)$ is the maximum power that is allowed to be allocated by user $q$ for the frequency bin $k$.
\end{itemize}

Each receiver estimates the channel between itself and all the transmitters, which is private information. The power allocation vectors are public information, i.e. known to all users. Each receiver computes the optimal power allocation across the frequency bins for its own link and transmits it back to the corresponding transmitter in a low bit-rate error-free feedback channel. Note that this leads to sharing of more information compared to other works in literature such as \cite{jrnl:AIWFA}. The channel state information estimated by each receiver is assumed to have a bounded uncertainty of unknown distribution. Ellipsoid is often used to approximate unknown and potentially complicated uncertainty sets \cite{bk:cvxBoyd}. The ellipsoidal approximation has the advantage of parametrically modelling complicated data sets and thus provides a convenient input parameter to algorithms. Further, in certain cases, there are statistical reasons leading to ellipsoidal uncertainty sets and also results in optimization problems with convenient analytical structures \cite{jrnl:ellipsoid1, jrnl:ellipsoid2}.

We consider that at each frequency, the uncertainty in the channel state information of each user is deterministically modelled under an ellipsoidal approximation\footnote{More specifically, we have a spherical approximation in \eqref{eq:deltaF}.}
\begin{equation}\label{eq:deltaF}
\begin{split}
    \mathcal{F}_q = \bigg\{ F_{rq}(k) + \Delta F_{rq,k} \ : \ \sum_{r \neq q} | \Delta & F_{rq,k} |^2 \leq \epsilon_q^2 \\ &\forall k=1,\dots,N \bigg\},
\end{split}
\end{equation}
where $\epsilon_q \geq 0 \ \forall \ q \in \Omega$ is the uncertainty bound and 
\begin{equation}\label{eq:defF}
    F_{rq}(k) \triangleq \frac{|H_{rq}(k)|^2}{|H_{qq}(k)|^2} ,
\end{equation}
with $F_{rq}(k)$ being the nominal value. We can consider uncertainty in $F_{rq}(k)$ instead of $H_{rq}(k)$ because a bounded uncertainty in $F_{rq}(k)$ and $H_{rq}(k)$ are equivalent, but with different bounds.\footnote{The model considered here has some redundancy in the uncertainty for the case when $F_{rq}(k)=0$ which leads to including $F_{rq}(k) + \Delta F_{rq,k} < 0$ in the model which can never happen in practice. However, this does not affect the solution in our method due to the nature of the max-min problem formulation in \eqref{eq:WCp1} which leads to selection of positive values of $\Delta F_{rq,k}$.}

The information rate of user $q$ can be written as \cite{bk:ThomasCover}
\begin{equation}\label{eq:rate}
    R_q =  \sum_{k=1}^{N} \log \bigg( 1 + \frac{p_q(k)}{\sigma_q^2(k)+\sum_{r \neq q} F_{rq}(k) p_r(k)} \bigg),
\end{equation}
where $\sigma_q^2(k) \triangleq \bar{\sigma}_q^2(k)/|H_{qq}(k)|^2$.
The two popular measures of ``inefficiency'' of the equilibria of a game are the price of anarchy and the price of stability. The price of anarchy is defined as the ratio between the objective function value at the socially optimal solution and the \emph{worst} objective function value at any equilibrium of the game \cite{bk:nisan2007algorithmic}. The price of stability is defined as the ratio between the objective function value at the socially optimal solution and the \emph{best} objective function value at any equilibrium of the game \cite{bk:nisan2007algorithmic}. We consider the sum-rate of the system as the social objective function. The sum-rate of the system is given by
\begin{equation}\label{eq:sum_rate}
    S = \sum_{q=1}^{Q} R_q.
\end{equation}
In our case, the price of stability and anarchy are the same as we prove the sufficient conditions for the existence of a unique equilibrium in Theorem~\ref{thrm:B_existNE}. Thus, the price of anarchy, $\PoA$, is the ratio of the sum-rate of the system at the social optimal solution, $S^*$, and the sum-rate of the system at the robust-optimization equilibrium, $S^{\rob}$, i.e.,
\begin{equation}\label{eq:PoA}
    \PoA = \frac{S^{*}}{S^{\rob}\ }.
\end{equation}
Note that a lower price of anarchy indicates that the equilibrium is more efficient.
\section{Robust Rate-Maximization Game Formulation}\label{sec:RobGameForm}
\subsection{Nominal Game - No Channel Uncertainty}
The problem of power allocation across the frequency bins is cast as a strategic noncooperative game with the SISO links as players and their information rates as pay-off functions \cite{jrnl:AIWFA}.
Mathematically, the nominal game $\mathscr{G}^{\nom}$ can be written as, $\quad \forall q \in \Omega,$
\begin{equation}\label{eq:Nom_game}
    \left. \begin{aligned}
\max_{\mathbf{p}_q} \quad &  \sum_{k=1}^{N} \log\bigg( 1 + \frac{p_q(k)}{\sigma_q^2(k)+\displaystyle\sum_{r \neq q}^{} F_{rq}(k) p_r(k)} \bigg) \\
\subto \quad & \mathbf{p}_q \in \mathcal{P}_q,
\end{aligned} \right.
\end{equation}
where $\Omega \triangleq \{1, \dots, Q\}$ is the set of the $Q$ players (i.e.\ the SISO links) and $\mathcal{P}_q$ is the set of admissible strategies of user $q$, which is defined as
\begin{equation}\label{eq:S_admSet}
\begin{split}
    \mathcal{P}_q  \triangleq \bigg\{ \mathbf{p}_q \in \mathbb{R}^{N} &:\ 0 \leq p_q(k) \leq p_q^{max}(k), \\
    & \sum_{k=1}^{N} p_q(k) = P_q, \quad k=1, \dots, N \bigg\}.
\end{split}
\end{equation}
The inequality constraint in \eqref{eq:S_cons1} is replaced with the equality constraint in \eqref{eq:S_admSet} as, at the optimum of each problem in \eqref{eq:Nom_game}, the constraint must be satisfied with equality \cite{jrnl:AIWFA}. To avoid the trivial solution $p_q(k)=p_q^{max}(k) \ \forall k$, it is assumed that $\sum_{k=1}^N p_q^{max} > P_q$. Further, the players can be limited to pure strategies instead of mixed strategies, as shown in \cite{jrnl:Opt1}.
\subsection{Robust Game - With Channel Uncertainty}
According to the robust game model \cite{jrnl:robGameTh}, each player
formulates a best response as the solution of a robust (worst-case) optimization problem for the uncertainty in the
payoff function (information rate), given the other players' strategies. If all the players know that everyone else is using the
robust optimization approach to the payoff uncertainty, they would then be able to mutually predict each other's behaviour. The robust game $\mathscr{G}^{\rob}$ where each player $q$ formulates a worst-case robust optimization problem can be
written as, $\ \forall \ q \in \Omega$,
\begin{equation}\label{eq:WCp0}
  \left. \begin{aligned}
\max_{\mathbf{p}_q}& \ \min_{\tilde{F}_{rq} \in \mathcal{F}_q} \  \sum_{k=1}^{N} \log\bigg( 1 + \frac{p_q(k)}{\sigma_q^2(k)+\sum\limits_{r \neq q} \tilde{F}_{rq}(k) p_r(k)} \bigg) \\
\subto &\quad \mathbf{p}_q \in \mathcal{P}_q,
\end{aligned} \right.
\end{equation}
where $\mathcal{F}_q$ is the uncertainty set which is modelled under ellipsoid approximation as shown in \eqref{eq:deltaF}. This optimization problem using uncertainty sets can be equivalently written in a form represented by protection functions \cite{conf:Infocom08_DRO} as, $\ \forall q \in \Omega,$
\begin{align}\nn
\max_{\mathbf{p}_q}& \!  \min_{\Delta F_{rq,k}}  \sum_{k=1}^{N}
\log\bigg( 1 \!+ \!\frac{p_q(k)}{\sigma_q^2(k)\!+\!\!\sum\limits_{r \neq q}  (F_{rq}\!(k)\! +\! \Delta F_{rq,k}) p_r(k)} \bigg) \\ 
 \subto &\quad  \sum_{r \neq q} | \Delta F_{rq,k} |^2 \leq
\epsilon_q^2, \quad \mathbf{p}_q \in \mathcal{P}_q.\label{eq:WCp1}
\end{align}
From the Cauchy-–Schwarz inequality \cite{bk:MatrixAnalysis}, we get
\begin{equation}\label{eq:CS1}
\begin{split}
    \sum_{r \neq q}^{} \Delta F_{rq,k} p_r(k) &\leq \bigg[ \sum_{r \neq q}^{} |\Delta F_{rq,k}|^2 \sum_{r \neq q}^{} |p_r(k)|^2 \bigg]^\frac{1}{2}\\
    & \leq \epsilon_q \sqrt{\textstyle\sum_{r \neq q} p_r^2(k)}
\end{split}
\end{equation}
Using \eqref{eq:CS1}, we get the robust game $\mathscr{G}^{\rob}$ as, $\ \forall q \in \Omega,$
\begin{align}\nn
\max_{\mathbf{p}_q}&  \sum_{k=1}^{N} \log\bigg(1 \!+  \!
\frac{p_q(k)}{\sigma_q^2(k)+\!\!\sum\limits_{r \neq q} F_{rq}(k) p_r(k)\!  + \epsilon_q \sqrt{\sum\limits_{r \neq q} p_r^2(k)}} \bigg) \\
\subto & \quad  \mathbf{p}_q \in \mathcal{P}_q.\label{eq:WCp2}
\end{align}

Now that we have defined the problem for robust rate-maximization under bounded channel uncertainty, we present the solution to the optimization problem in \eqref{eq:WCp2} for a single-user in the following section.
\subsection{Robust Waterfilling Solution}\label{sec:RWF}
The closed-form solution to the robust optimization problem in \eqref{eq:WCp2} for any particular user $q$ is given by the following theorem:
\begin{lemma}\label{thrm:SU_WF}
Given the set of power allocations of other users $\mathbf{p}_{-q} \triangleq \{\mathbf{p}_{1}, \dots, \mathbf{p}_{q-1}, \mathbf{p}_{q+1}, \dots, \mathbf{p}_{Q}\}$, the solution to 
\begin{align}\nn
\max_{\mathbf{p}_q}&  \sum_{k=1}^{N} \log\bigg(1 \!+  \!
\frac{p_q(k)}{\sigma_q^2(k)+\!\!\sum\limits_{r \neq q} F_{rq}(k)
p_r(k)\!  + \epsilon_q \sqrt{\sum\limits_{r \neq q} p_r^2(k)}} \bigg) \\
\subto & \quad  \mathbf{p}_q \in \mathcal{P}_q.\label{eq:P1}
\end{align}
is given by the waterfilling solution
\begin{equation}\label{eq:P1soln}
    \mathbf{p}_q^\star=\WFq(\mathbf{p}_{-q}),
\end{equation}
where the waterfilling operator $\WFq(\cdot)$ is defined as $\big[\WFq (\mathbf{p}_{-q}) \big]_k \triangleq $
\begin{equation}\label{eq:WF}
\Bigg[ \mu_q - \sigma_q^2 (k)- \displaystyle\sum_{r \neq q} F_{rq}(k) p_r(k)
-\epsilon_q \sqrt{\sum_{r \neq q} p_r^2(k)}
\Bigg]_0^{p_q^{max}(k)}
\end{equation}
for $ k=1, \dots, N,$ where $\mu_q$ is chosen to satisfy the power constraint $\sum_{k=1}^N p_q^\star(k)=P_q$.
\end{lemma}
\begin{proof}
This can be shown using the Karush-Kuhn-Tucker optimality conditions \cite{bk:cvxBoyd} of this problem.
\end{proof}

The robust waterfilling operation for each user is a distributed worst-case optimization under bounded channel uncertainty. Compared with the original waterfilling operation in \cite{jrnl:AIWFA} under perfect CSI (i.e. $\epsilon_q\equiv 0$), we see that an additional term has appeared in \eqref{eq:WF} for $\epsilon_q>0$.
This additional term can be interpreted as a penalty for allocating power to frequencies having a large product of uncertainty bound and norm of the powers of the other players currently transmitting in those frequencies. This is because the users assume the worst-case interference from other users and are thus conservative about allocating power to such channels where there is a strong presence of other users.

Having derived the robust waterfilling solution for a single user in the presence of channel uncertainty, we consider whether a stable equilibrium for the system exists and if so, what its properties are and how it can be computed in the multi-user scenario in the following section.
\section{Robust-Optimization Equilibrium}\label{sec:RoE}
The solution to the game $\mathscr{G}^{\rob}$ is the robust-optimization equilibrium (RE). At any robust-optimization equilibrium of this game, the optimum action profile of the players $\{ \mathbf{p}_q^\star \}_{q \in \Omega}$ must satisfy the following set of simultaneous waterfilling equations: $\forall q \in \Omega$,
\begin{equation}\label{eq:S_nash}
    \mathbf{p}_q^\star = \WFq( \mathbf{p}^\star_{1}, \dots, \mathbf{p}^\star_{q-1}, \mathbf{p}^\star_{q+1}, \dots, \mathbf{p}^\star_{Q}) = \WFq(\mathbf{p}^\star_{-q}).
\end{equation}

It can easily be verified that the RE reduces to the Nash equilibrium of the system \cite{jrnl:AIWFA} when there is no uncertainty in the system. In Section~\ref{sec:Efficiency}, we analyse the global efficiency of the RE and show that the RE has a higher efficiency than the Nash equilibrium due to a penalty for interference which encourages better partitioning of the frequency space among the users. The robust asynchronous iterative waterfilling algorithm for computing the RE of game $\mathscr{G}^{\rob}$ in a distributed fashion is described in Algorithm~\ref{al:RIWFA}.
\begin{algorithm}[t]
\caption{-- Robust Iterative Waterfilling Algorithm}\label{al:RIWFA}
  \begin{algorithmic}
    \STATE \textbf{Input:}\\
    $\Omega$: Set of users in the system\\
    $\mathcal{P}_q$: Set of admissible strategies of user $q$\\
    $\mathcal{T}_q$: Set of time instants $n$ when the power vector $\mathbf{p}_q^{(n)}$ of user $q$ is updated\\
    $T$: Number of iterations for which the algorithm is run\\
    $\tau_r^q(n)$: Time of the most recent power allocation of user $r$ available to user $q$ at time $n$\\
    $\WFq(\cdot)$: Robust waterfilling operation in \eqref{eq:WF}
    \STATE \textbf{Initialization:} $n=0$ and $\mathbf{p}_q^{(0)} \leftarrow$ any $\mathbf{p} \in \mathcal{P}_q, \ \forall q \in \Omega$
    \FOR{$n=0$ to $T$}
    \STATE $\mathbf{p}_q^{(n+1)} = \Bigg\{  \begin{aligned}
    &\WFq\Big(\mathbf{p}_{-q}^{( \tau^{q}(n) )}\Big), \  & \text{if } n \in \mathcal{T}_q, \\
    &\mathbf{p}_q^{(n)},                    & \text{otherwise,}
    \end{aligned}   \quad \forall q \in \Omega$.
    \ENDFOR
  \end{algorithmic}
\end{algorithm}
\subsection{Analysis of the RE of Game $\mathscr{G}^{\rob}$}
Let $\mathcal{N} = \{1, \dots, N\}$ be the set of frequency bins. Let $\mathcal{D}_q^\circ$ denote the set of frequency bins that user $q$ would never use as the best response to any set of strategies adopted by the other users,
\begin{equation}\label{eq:S_setD}
\begin{split}
    \mathcal{D}_q^\circ \triangleq \Big\{ k \in \{& 1, \dots, N\} \ :  \\ &\big[\WFq(\mathbf{p}_{-q}) \big]_k = 0 \quad \forall\mathbf{p}_{-q} \in \mathcal{P}_{-q} \Big\},
\end{split}
\end{equation}
where $\mathcal{P}_{-q} \triangleq \mathcal{P}_1 \times \dots \times \mathcal{P}_{q-1} \times \mathcal{P}_{q+1} \times \dots \times \mathcal{P}_Q$. The non-negative matrices $\mathbf{E}$ and $\mathbf{S}^{max} \in \mathbb{R}_+^{Q \times Q}$ are defined as
\begin{equation}\label{eq:E_defn}
    [\mathbf{E}]_{qr} \triangleq \left\{ \begin{aligned}
    &\epsilon_q, \quad &\text{if } r \neq q, \\
    &0,      &\text{otherwise}, \end{aligned} \right.
\end{equation}
and
\begin{equation}\label{eq:S_max}
    [\mathbf{S}^{max}]_{qr} \triangleq \left\{ \begin{aligned}
    &\max_{k \in \mathcal{D}_q \cap \mathcal{D}_r}  F_{rq}(k), \quad &\text{if } r \neq q, \\
    &0,      &\text{otherwise}, \end{aligned} \right.
\end{equation}
where $\mathcal{D}_q$ is any subset of $\{1, \dots, N\}$ such that $\mathcal{N-D}_q^\circ \subseteq \mathcal{D}_q \subseteq \{1, \dots, N\}$.

The sufficient condition for existence and uniqueness of the RE of game $\mathscr{G}^{\rob}$ and for the guaranteed convergence of Algorithm~\ref{al:RIWFA} is given by the following theorem:
\begin{thrm}\label{thrm:B_existNE}
Game $\mathscr{G}^{\rob}$ has at least one equilibrium for any set of channel values and transmit powers of the users. Furthermore, the equilibrium is unique and the asynchronous iterative waterfilling algorithm described in Algorithm~\ref{al:RIWFA} converges to the unique RE of game $\mathscr{G}^{\rob}$ as the number of iterations for which the algorithm is run, $T \rightarrow \infty$ for any set of feasible initial conditions if
\begin{equation}\label{eq:B_existNE}
\rho(\mathbf{S}^{max})<1-\rho(\mathbf{E}),
\end{equation}
where $\mathbf{E}$ and $\mathbf{S}$ are as defined in \eqref{eq:E_defn} and \eqref{eq:S_max} respectively.
\end{thrm}
\begin{proof}
See Appendix~\ref{appendix:B_existNE}.
\end{proof}

In the absence of uncertainty, i.e. when $\epsilon_q=0 \ \forall q \in \Omega$, we can see that this condition reduces to condition (C$1$) in \cite{jrnl:AIWFA} as expected. Since $\rho(\mathbf{E}) \geq 0$, the condition on $\mathbf{S}^{max}$ becomes more stringent as the uncertainty bound increases, i.e. the set of channel coefficients for which the existence of a unique equilibrium and the convergence of the algorithm is guaranteed shrinks as the uncertainty bound increases. Also, from Lemma~\ref{lemma:WFcontract}, we can see that the modulus of the waterfilling contraction increases as uncertainty increases. This indicates that the convergence of the iterative waterfilling algorithm becomes slower as the uncertainty increases, as shown in simulation results.
\begin{corr}\label{corr:Equale}
When the uncertainties of all the $Q$ users are equal (say $\epsilon$), the RE of the game $\mathscr{G}^{\rob}$ is unique and Algorithm~\ref{al:RIWFA} converges to the unique RE of game $\mathscr{G}^{\rob}$ as $T \rightarrow \infty$ for any set of feasible initial condition if
\begin{equation}\label{eq:NE_corr}
    \rho(\mathbf{S}^{max}) < 1- \epsilon(Q-1)
\end{equation}
\end{corr}
\begin{proof}
When the uncertainties of all $Q$ users is $\epsilon$, we get $\rho(\mathbf{E})=\epsilon(Q-1)$.
\end{proof}

The above corollary explicitly shows how the uncertainty bound and the number of users in the system affect the existence of the equilibrium and the convergence to the equilibrium using an iterative waterfilling algorithm.  For a fixed uncertainty bound, as the number of users in the system increases, there is a larger amount of uncertain information in the system. Hence, the probability that a given system for a fixed uncertainty bound will converge will decrease as the number of users in the system increases. Also, if $\epsilon(Q-1) \geq 1$, we will not have a guaranteed unique equilibrium and algorithmic convergence for non-zero uncertainty bounds regardless of the channel coefficients. This will help designers plan for the appropriate uncertainty bounds based on the planned number of users in the system.
\section{Efficiency at the Equilibrium -- Two-User Case}\label{sec:Efficiency}
In this section, we analyse  the effect of uncertainty on the social output of the system.  For the two user case, the worst-case interference in each frequency reduces to $\big(F_{rq}(k)+\epsilon_q \big) p_r(k)$ with $q,r = 1,2$ and $q \neq r$. This means that the robust waterfilling operation for the two user case ($Q=2$) is simply the standard waterfilling solution with the worst-case channel coefficients. We restrict the analysis to the two-user case with identical noise variance $\sigma^2_q(k) = \sigma^2 \ \forall k,q$ across all frequencies and identical uncertainty bounds $\epsilon_1=\epsilon_2=\epsilon$ and total power constraints $\sum_{k=1}^N p_1(k) = \sum_{k=1}^N p_2(k) = P_T$ for both users. These results can be extended to the non-identical case along similar lines. In order to develop a clear understanding of the behaviour of the equilibrium, the sum-rate of the system is first analyzed for a system with two frequencies ($N=2$) and then extended to systems with large ($N \rightarrow \infty$) number of frequencies.
\subsection{Two Frequency Case ($N=2$)}
Consider a two-frequency anti-symmetric system as shown in Figure~\ref{fig:2sys_model} where the channel gains are $|H_{11}(1)|^2 =|H_{11}(2)|^2 =|H_{22}(1)|^2 =|H_{22}(2)|^2=1; |H_{12}(2)|^2=|H_{21}(1)|^2=\alpha$ and $|H_{12}(1)|^2=|H_{21}(2)|^2=m \alpha$ with $m>1$ and $0<\alpha<1$. From \eqref{eq:WF} the power allocations at the robust-optimization equilibrium of this system are,
\begin{equation}\label{eq:2pow_allocs1}
\begin{split}
p_1(1) &= [\mu_1 - \sigma^2 -(\alpha+\epsilon)p_2(1)]^+, \\
p_1(2) &= [\mu_1 - \sigma^2 -(m \alpha+\epsilon)p_2(2)]^+, \\
p_2(1) &= [\mu_2 - \sigma^2 -(m \alpha+\epsilon)p_1(1)]^+, \\
p_2(2) &= [\mu_2 - \sigma^2 -(\alpha+\epsilon)p_1(2)]^+,
\end{split}
\end{equation}
and the total power constraint for each user is $p_1(1)+p_1(2)=p_2(1)+p_2(2)=1$. Let $p_1(1)=p$, hence, by symmetry, $p_1(2)=p_2(1)=1-p$, $p_2(2)=p$ and $\mu_1=\mu_2=\mu$.
\begin{figure}[t]
 \centering
 \footnotesize
 \psfrag{Txa}[Bc]{Tx 1}\psfrag{Txb}[Bc]{Tx 2}\psfrag{Rxa}{Rx 1}\psfrag{Rxb}{Rx 2} \psfrag{CG1}{$[\alpha, m\alpha$]}\psfrag{CG2}{$[m\alpha, \alpha]$}\psfrag{g}[l]{[1,1]}\psfrag{h}[l]{[1,1]}
  \includegraphics[width=0.95\columnwidth]{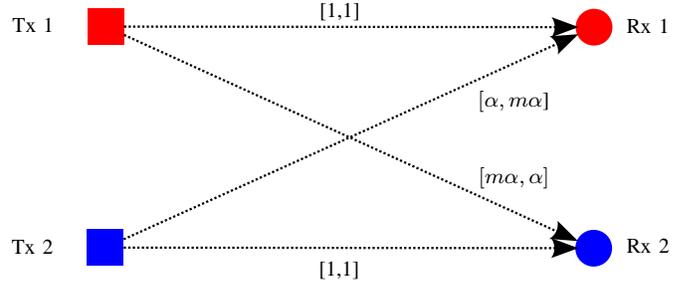}\\
  \caption{Anti-symmetric system with $Q=2$, $N=2, \epsilon_1=\epsilon_2=\epsilon$ and the noise variances for both users in both frequencies is $\sigma^2$. The channel gains are: $|H_{11}(1)|^2=|H_{11}(2)|^2=|H_{22}(1)|^2=|H_{22}(2)|^2=1; |H_{12}(2)|^2=|H_{21}(1)|^2=\alpha$ and $|H_{12}(1)|^2=|H_{21}(2)|^2=m \alpha$ with $m>1$ and $0 < \alpha<1$. The power allocations for this system at the robust-optimization equilibrium are presented in \eqref{eq:2pow_allocs1}.} \label{fig:2sys_model}
\end{figure}

The following theorem presents the effect of uncertainty on the sum-rate and price of anarchy of the system for the high interference and low interference cases.
\begin{thrm}\label{thrm:2F}
For the two-user two-frequency anti-symmetric system described above, we have the following results:
\begin{itemize}
  \item High interference: When $\sigma^2 \ll \alpha(1-p)$, the sum-rate increases and the price of anarchy decreases as the channel uncertainty bound increases.
  \item Low interference: When $\sigma^2 \gg m\alpha p$, the sum-rate decreases and the price of anarchy increases as the channel uncertainty bound increases.
\end{itemize}
\end{thrm}
\begin{proof}
See Appendix~\ref{appendix:2F}.
\end{proof}

We can see from this result that the RE behaves in opposite ways when there is high interference and when there is low interference in the system. This suggests that there might be a certain level of interference where the sum-rate and price of anarchy do not change with change in uncertainty. This is given by the following proposition:
\begin{prop}\label{prop:crit}
At the level of interference $\alpha=\alpha_{\crit}$, where
      \begin{equation}\label{eq:alpha_crit}
         \alpha_{\crit} = \frac{\sigma^2}{2m}\left( \left( (m+1)^2 + \frac{4m}{\sigma^2} \right)^{\frac{1}{2}} - m - 1 \right),
      \end{equation}
the sum-rate and the price of anarchy are independent of the level of uncertainty or the power allocation used. Furthermore, at this value of interference, the price of anarchy is equal to unity.
\end{prop}
\begin{proof}
See Appendix~\ref{crit_proof}
\end{proof}

We can see that even for such a simple system, the global behaviour of the robust-optimization equilibrium appears to be quite complex. This indicates that the global properties of the robust-optimization equilibrium for larger systems is quite strongly dependent on the level of interference in the system, which is seen in the following results with asymptotic number of frequencies. However, the underlying nature of the algorithm for the two-frequency case is that the RE moves towards the FDMA solution as the uncertainty bound increases (from \eqref{eq:2dp_de}).
\subsection{Large Number of Frequencies ($N \rightarrow \infty$)}
We consider $J(k)$, defined as
\begin{equation}\label{eq:defJ}
    J(k) \triangleq -p_1(k)p_2(k),
\end{equation}
as a measure of the extent of partitioning of the frequency $k$. It is minimum ($J(k)=-1$) when both the users allocate all their total power to the same frequency $k$ and is maximum ($J(k)=0$) when at most one user is occupying the frequency $k$. Note that $J(k) = 0 \ \forall k \in \{1,\dots,N\}$ when the users adopt an FDMA scheme.

The following lemma describes the effect of the uncertainty bound on the extent of partitioning of the system:
\begin{lemma}\label{lemma:degPart}
When the number of frequencies, $N \rightarrow \infty$, the extent of partitioning in every frequency is non-decreasing as the uncertainty bound of the system increases for any set of channel values, i.e.,
\begin{equation}\label{eq:lemma_degPart}
    \D J(k) \geq 0 \quad \forall k \in \{1,\dots,N\} \ \ \text{when} \ \ N \rightarrow \infty,
\end{equation}
with equality for frequencies where $J(k)=0$, where $J(k)$ is defined in \eqref{eq:defJ}
\end{lemma}
\begin{proof}
See Appendix~\ref{appendix:degPart}.
\end{proof}

The above lemma suggests that the robust-optimization equilibrium moves towards greater frequency-space partitioning as the uncertainty bound increases when there is a large number of frequencies in the system. In other words, the RE is moving closer to an FDMA solution under increased channel uncertainty. When the FDMA solution is globally optimal, this will lead to an improvement in the performance of the equilibrium. This is stated in the following theorem:
\begin{thrm}\label{thrm:LargeN}
As the number of frequencies, $N \rightarrow \infty$, the sum-rate (price of anarchy) at the robust-optimization equilibrium of the system is non-decreasing (non-increasing) as the uncertainty bound increases if, $\forall \ k\in \{1,\dots,N\}$,
\begin{equation}\label{eq:sumRateTh}
    (F_{21}(k) - \epsilon) ( F_{12}(k) - \epsilon) > \frac{1}{4}. 
\end{equation}
\end{thrm}
\begin{proof}
Using \cite[Corollary 3.1]{jrnl:ZQLuo}, we find that the sum of the rates of the two users in the frequency $k$ is quasi-convex only if $F_{21}(k) F_{12}(k) > 1/4$. Let $C$ be the minimum number of frequencies occupied by any user. When there are only two users and a large number of frequencies, $C \gg 1$. If the condition $F_{21}(k)  F_{12}(k) > \frac{1}{4}(1 + \frac{1}{C-1})^2$ is satisfied for some $C \geq 2$ for all frequencies $k \in \{1,\dots,N\}$ (thus satisfying $F_{21}(k) F_{12}(k) > 1/4$), then the Pareto optimal solution is FDMA \cite[Theorem 3.3]{jrnl:ZQLuo}. This needs to be satisfied for the worst-case channel coefficients which leads to \eqref{eq:sumRateTh}. Thus, the solution moving closer to FDMA will improve the sum-rate of the system. From Lemma~\ref{lemma:degPart}, the robust equilibrium moves closer to FDMA as uncertainty increases and thus will result in an improvement in sum-rate.

The Pareto optimal solution under this condition (which is FDMA) is constant under varying uncertainty bounds as such an uncertainty in the interference coefficients $F_{12}(k)$ and $F_{21}(k)$ does not affect the FDMA solution where there is no interference. Thus, an increase in sum-rate will result in an decrease in price of anarchy.
\end{proof}
\begin{Remark}
For the special case of frequency-flat systems, at the equilibrium, all users have equal power allocation to all frequencies and this is not dependent on the uncertainty in the CSI. This leads to no change in the extent of partitioning and thus sum-rate and price of anarchy are not affected by uncertainty.
\end{Remark}
\begin{Remark}
The results of this section are not just limited to the robust-optimization equilibrium for the system presented here. When the uncertainty $\epsilon=0$, the framework presented here can be used to analyse the behaviour of the Nash equilibrium of the iterative waterfilling algorithm as a function of the interference coefficients. 
\end{Remark}
\begin{Remark}
The modified waterfilling operation in \eqref{eq:WF} can also be used as a pricing mechanism to achieve improved sum-rate performance in a system with no uncertainty where $\epsilon$ is a design parameter, with all the analytical results presented here still being valid.
\end{Remark}
%
%
\section{Simulation Results}\label{sec:results}
\begin{figure}\centering
\subfloat[a][Sum-rate of the system under high interference vs. uncertainty and interference.]
{\label{fig:2sum_rate_high}\begin{minipage}{\columnwidth}
\centering
\includegraphics[width=\columnwidth]{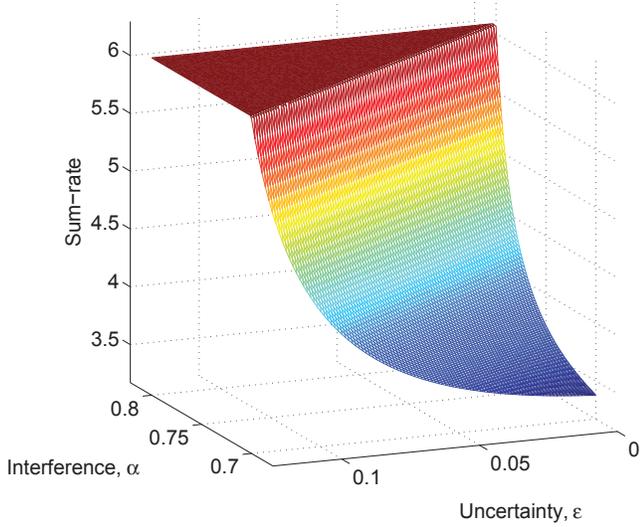}
\end{minipage}} \\
 \subfloat[b][Sum-rate of the system under low interference vs. uncertainty and interference.]
{\label{fig:2sum-rate_low}\begin{minipage}{\columnwidth}
\centering
\includegraphics[width=\columnwidth]{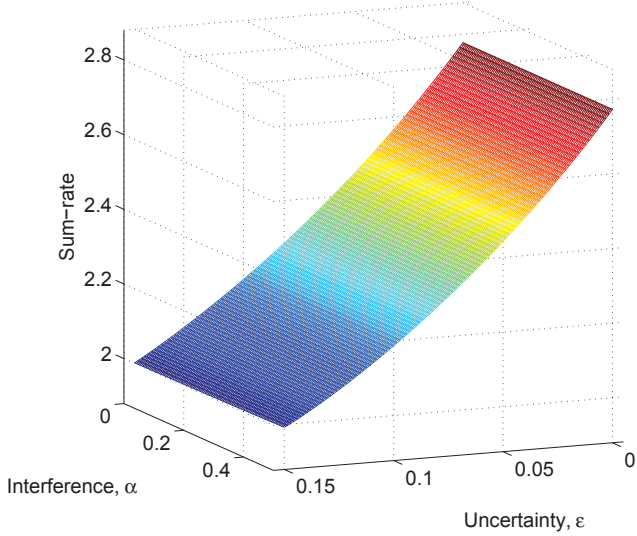}
\end{minipage} }
  \caption{Simulation results for the anti-symmetric system in Figure~\ref{fig:2sys_model}. Note that the zero uncertainty corresponds to the Nash equilibrium}\label{fig:2sims}
\end{figure}
\begin{figure}
\centering
\includegraphics[width=0.8\columnwidth]{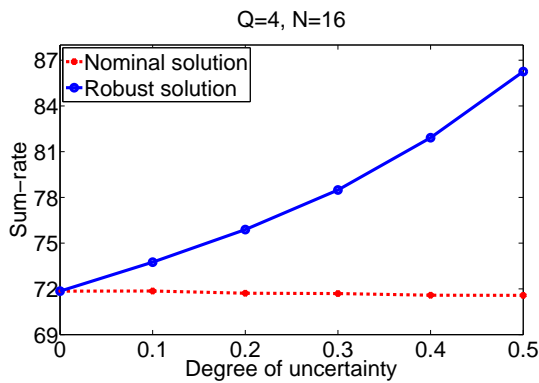}\\
\caption{Sum-rate of the system vs. uncertainty $\delta$.}
\label{fig:Delta_sum_rateSISO}
\end{figure}
\begin{figure}
\centering
\includegraphics[width=0.8\columnwidth]{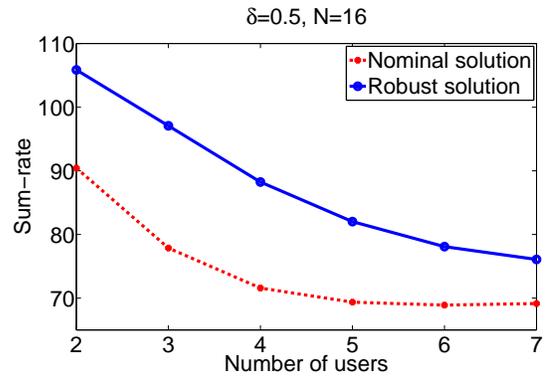}\\
\caption{Sum-rate of the system vs. number of users, $Q$.}\label{fig:Q_sum_rateSISO}
\end{figure}
\begin{figure}
\centering
\includegraphics[width=0.8\columnwidth]{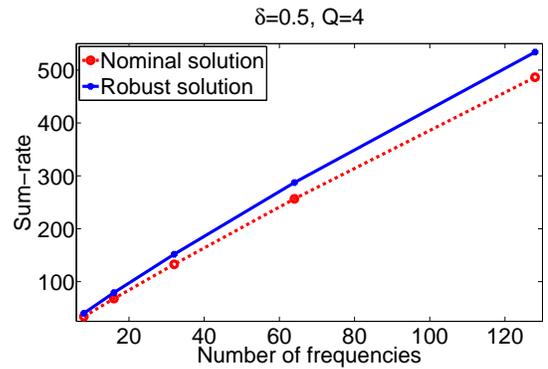}\\
\caption{Sum-rate of the system vs. number of frequencies, $N$.}\label{fig:Freq_sum_rateSISO}
\end{figure}
\begin{figure}
\centering
\includegraphics[width=0.8\columnwidth]{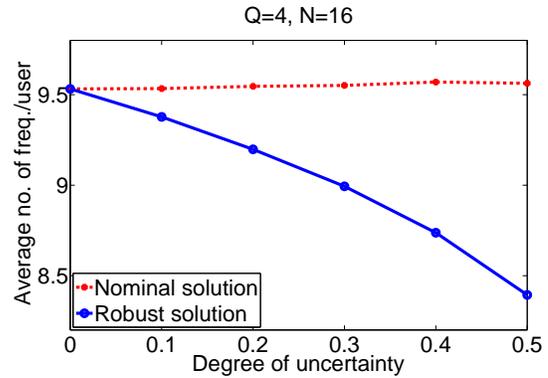}\\
\caption{Average number of channels occupied per user vs. uncertainty, $\delta$.}\label{fig:Delta_ch_occSISO}
\end{figure}
\begin{figure}
\centering
\includegraphics[width=0.8\columnwidth]{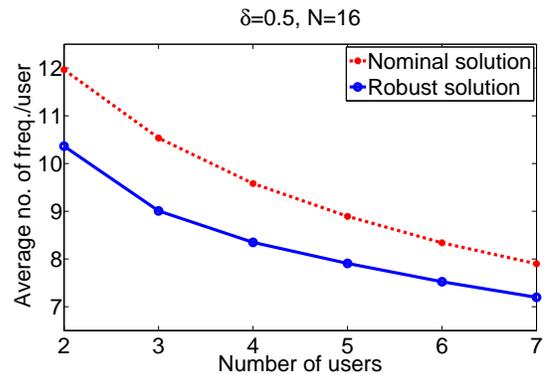}\\
\caption{Average number of channels occupied per user vs. number of users, $Q$.}\label{fig:Q_ch_occSISO}
\end{figure}
\begin{figure}[th]
\centering
\includegraphics[width=0.8\columnwidth]{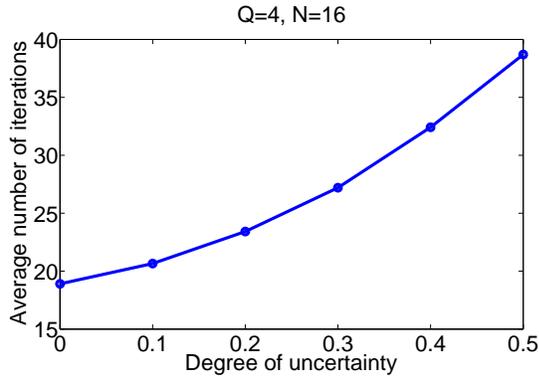}\\
\caption{Average number of iterations vs. uncertainty, $\delta$.}\label{fig:Delta_iterSISO}
\end{figure}
\begin{figure}[th]
\centering
\includegraphics[width=0.8\columnwidth]{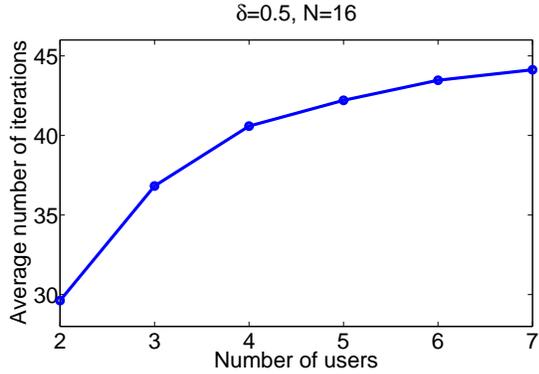}\\
\caption{Average number of iterations vs. number of users, $Q$.}\label{fig:Q_iter_SISO}
\end{figure}
\begin{figure}[th]
\centering
\includegraphics[width=0.8\columnwidth]{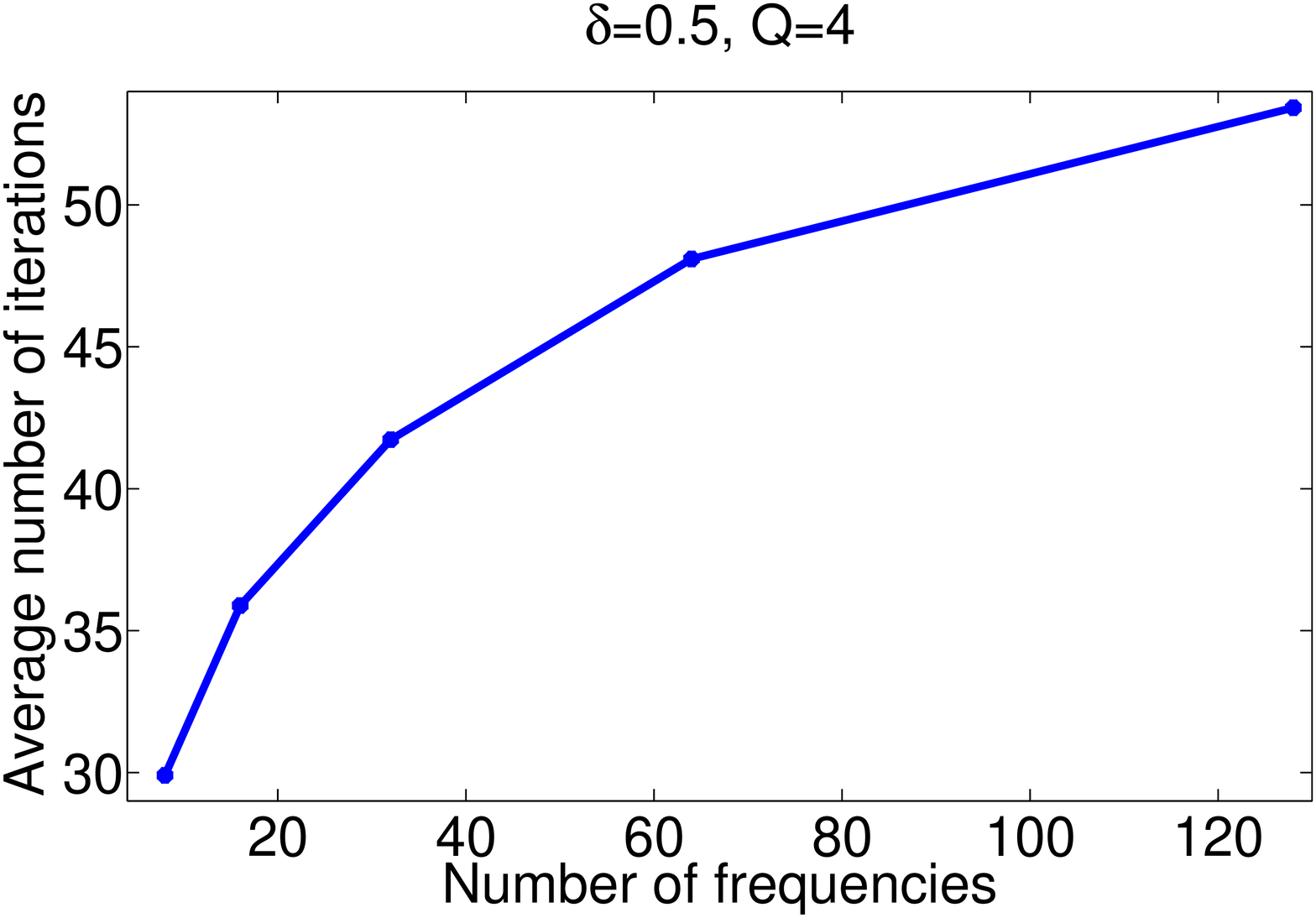}\\
\caption{Average number of iterations vs. number of frequencies, $N$.}\label{fig:Freq_iterSISO}
\end{figure}
In this section, we present some simulation results to study the impact of channel uncertainty on the RE by comparing it with the ideal scenario of NE under perfect CSI. Figure~\ref{fig:2sims} shows the simulation results for the two user and two frequency scenario and Figures~\ref{fig:Delta_sum_rateSISO}--\ref{fig:Freq_iterSISO} show the results for a more general system under different settings.

In Figure~\ref{fig:2sum_rate_high}, we can see the sum-rate at high interference as a function of interference and uncertainty in the system shown in Figure~\ref{fig:2sys_model}. The flat region corresponds to the sum-rate at Pareto optimal solution (FDMA) and the edge of the surface corresponds to the sufficient condition in \eqref{eq:B_existNE}. It can be seen that the Nash equilibrium (uncertainty, $\epsilon=0$) moves closer to the Pareto optimal solution as the interference increases. It is also evident that the sum-rate increases for a fixed interference as uncertainty increases, as expected from Theorem~\ref{thrm:2F}. In Figure~\ref{fig:2sum-rate_low}, we can see the sum-rate at low interference as a function of interference and uncertainty. As expected from Theorem~\ref{thrm:2F}, the sum-rate decreases as the uncertainty increases.

In Figures~\ref{fig:Delta_sum_rateSISO}--\ref{fig:Freq_iterSISO}, the behaviour of the equilibrium under varying uncertainty bounds is investigated through numerical simulations. The simulations are computed for a system with $Q$ users and $N$ frequencies averaged over 5000 channel realizations. The channel gains are $H_{rq}(k) \sim N_C(0,1)$ for $r\neq q$ and $H_{qq}(k)\sim N_C(0,2.25)$. The channel uncertainty model has nominal value $F_{rq}(k)= F^{\mbox{\tiny ~true}}_{rq}(k)(1 + e_{rq}(k))$ with $e_{rq}(k) \sim U(-\frac{\delta}{2}, \frac{\delta}{2}), \ \delta<1$. The specific parameter values used for the simulations are provided above each figure.
Note that the zero uncertainty solution corresponds to the Nash equilibrium and the nominal solution is the solution resulting from using the erroneous channel values in the traditional rate-maximization game $\mathscr{G}^S$ without accounting for its uncertainty. The effect of uncertainty, number of users and number of frequencies on the average sum-rate of the system, the average number of frequencies occupied by each user and the average number of iterations for convergence are examined. In these figures, the Nash equilibrium point is when the uncertainty is zero.

In Figure~\ref{fig:Delta_sum_rateSISO}, it can be observed that the sum-rate at the Nash equilibrium under perfect CSI is less than the sum-rate at the robust-optimization equilibrium under imperfect CSI and that the gap in performance increases as the uncertainty $\delta$ increases. Under imperfect CSI, the power allocation using the robust-optimization equilibrium\footnote{\eqref{eq:WF} and \eqref{eq:S_nash} are in terms of absolute uncertainty $\epsilon$ while the simulations use relative uncertainty $\delta$. They are equivalent to one another.} in \eqref{eq:WF} and \eqref{eq:S_nash} has higher sum-rate as uncertainty increases, because the users are more cautious about using frequencies with significant interference, thus reducing the total amount of interference in the system. The Pareto optimal solution in this scenario is some form of FDMA (where the specific channels allocated to a particular user will depend on the actual channel gains), and does not depend on the channel uncertainty bound for a given system realization. Thus, when the sum-rate of the system increases with rise in uncertainty, we can expect that the price of anarchy will decrease.

In Figure~\ref{fig:Q_sum_rateSISO}, it can be observed that the sum-rate of the system under the robust solution reduces when the number of users increases. This is because having a greater number of users results in higher interference for all users and this effect is strong enough to counter user diversity which would have resulted in higher sum-rates if the users were on an FDMA scheme. In Figure~\ref{fig:Freq_sum_rateSISO}, it can be observed that the sum-rate of the system improves with increase in number of frequencies and also that the robust solution continues to perform better than the nominal solution even when the number of frequencies increases.

In Figure~\ref{fig:Delta_ch_occSISO}, it can be seen that the robust solution results in a lower average number of channels per user as the uncertainty, $\delta$ increases. Also, the total number of channels each user occupies at the robust-optimization equilibrium  is lesser than at the nominal solution, regardless of the number of users, as can be seen in Figure~\ref{fig:Q_ch_occSISO}.
This implies that the users are employing a smaller number of frequencies, which demonstrates the improved partitioning of the frequency space among the users in order to reduce interference. Hence, this leads to the higher sum-rates as observed in Figure~\ref{fig:Delta_sum_rateSISO}.

In Figures~\ref{fig:Delta_iterSISO} and~\ref{fig:Q_iter_SISO}, it can be observed that the average number of iterations for convergence increases as the uncertainty $\delta$, and the number of users $Q$ increase respectively. This is as expected from Lemma~\ref{lemma:WFcontract} and Corollary~\ref{corr:Equale}, as the modulus of the block-contraction in \eqref{eq:contractWF} increases as the uncertainty increases. This indicates that the step size of each iteration reduces as uncertainty increases, leading to slower convergence. In Figure~\ref{fig:Freq_iterSISO}, it can be observed that the average number of iterations increases with the number of frequencies. This is due to the fact that when there are more frequencies, and as such there is a greater probability of a channel realization being drawn from the tail of the Gaussian distribution used to generate them, which results in a smaller modulus of the block-contraction in \eqref{eq:contractWF} on average, a greater number of iterations are required in order to converge to the equilibrium. Thus, the trade-off for robust solutions with higher sum-rates is in a higher number of iterations before convergence.
\section{Conclusions}\label{sec:concl}
In this paper, we have presented a novel approach for rate-maximization games under bounded channel state information uncertainty. We have introduced a distribution-free robust formulation for the rate-maximization game. The solution to this game has been shown to be a modified waterfilling operation. The robust-optimization equilibrium (RE) for this game has been presented and sufficient conditions for its existence, uniqueness and asymptotic convergence of the algorithm to the RE have been derived. For the two-user case, the effect of uncertainty on the social output of the system has been analyzed. We analytically prove that the RE moves towards an FDMA solution as the uncertainty bound increases, when the number of frequencies in the system becomes asymptotically large. Thus, an interesting effect of improvement in sum-rate as the uncertainty bound increases is observed. In summary, for systems with significant interference, bounded channel uncertainty leads to an improved sum-rate but at the cost of greater number of iterations. This framework can be extended to MIMO rate-maximization games, cognitive radio with various interference constraints and other noncooperative games.
\appendix           
\subsection{Robust Waterfilling as a Projection Operation}
Let $\Phi_q(k)$ represent the denominator terms in \eqref{eq:P1}, which is the worst-case noise+interference
\begin{equation}\label{eq:Phi}
    \Phi_q(k) \triangleq \sigma_q^2(k)+\displaystyle\sum_{r \neq q} F_{rq}(k) p_r(k) + \epsilon_q \sqrt{\sum_{r \neq q} p_r^2(k)}.
\end{equation}
It has been shown in \cite{jrnl:Opt2} that the waterfilling operation can be interpreted as the Euclidean projection of a vector onto a simplex. Using this framework, the robust waterfilling operator in \eqref{eq:WF} can be expressed as the Euclidean projection of the vector $\boldsymbol{\Phi}_q \triangleq [ \Phi_q(1), \dots, \Phi_q(N) ]^T$ onto the simplex $\mathcal{P}_q$ defined in \eqref{eq:S_admSet}:
\begin{equation}\label{eq:WFproj}
    \WFq(\mathbf{p}_{-q}) = \big[ - \boldsymbol{\Phi}_q \big]_{\mathcal{P}_q },
\end{equation}
which can be conveniently written as
\begin{equation}\label{eq:WFprj}
    \WFq(\mathbf{p}_{-q}) = \bigg[ -\boldsymbol{\sigma}_q - \sum_{r \neq q} \mathbf{F}_{rq} \mathbf{p}_r - \epsilon_q \mathbf{f}_q  \bigg]_{\mathcal{P}_q },
\end{equation}
where
\begin{eqnarray}\label{eq:vecSigma}
    \boldsymbol{\sigma}_q &\triangleq& \left[ \sigma^2_q(1), \dots, \sigma^2_q(N) \right]^T ,\\
    \mathbf{F}_{rq} &\triangleq& \Diag \Big( F_{rq}(1), \dots, F_{rq}(N) \Big),\\
    \mathbf{f}_q &\triangleq& \bigg[ \sqrt{\textstyle\sum_{r \neq q} p_r^2(1)} \ , \dots, \sqrt{\textstyle\sum_{r \neq q} p_r^2(N)} \bigg]^T .
\end{eqnarray}
\subsection{Contraction Property of the Waterfilling Projection}\label{appndx:contract}
Given the waterfilling mapping $\WFQ(\cdot)$ defined as
\begin{equation}\label{eq:blKWF}
    \WFQ(\mathbf{p}) = (\WFq(\mathbf{p}_{-q}))_{q \in \Omega} \ : \ \mathcal{P} \mapsto \mathcal{P},
\end{equation}
where $\mathcal{P} \triangleq \mathcal{P}_1 \times \cdots \times \mathcal{P}_Q$, with $\mathcal{P}_q$ and $\WFq(\mathbf{p}_{-q})$ respectively defined in \eqref{eq:S_admSet} and \eqref{eq:WFprj}, the block-maximum norm is defined as\cite{bk:tsitsiklis1989parallel}
\begin{equation}\label{eq:WFbknorm}
    \big|\big|\WFQ(\mathbf{p})\big|\big|_{2,\block}^{\mathbf{w}} \triangleq \max_{q \in \Omega} \frac{\big|\big|\WFq(\mathbf{p}_q)\big|\big|_2}{w_q},
\end{equation}
where $\mathbf{w} \triangleq [w_1, \dots, w_Q]^T > \mathbf{0}$ is any positive weight vector. The vector weighted maximum norm is given by \cite{bk:MatrixAnalysis}
\begin{equation}\label{eq:vecNorm}
    || \mathbf{x} ||_{\infty,\VEC}^{\mathbf{w}} \triangleq \max_{q \in \Omega} \frac{|x_q|}{w_q}, \qquad \mathbf{w}>\mathbf{0}, \quad \mathbf{x}\in\mathbb{R}^{Q}
\end{equation}
The matrix weighted maximum norm is given by \cite{bk:MatrixAnalysis}
\begin{equation}\label{eq:matrixNorm}
    ||\mathbf{A}||_{\infty, \mat}^{\mathbf{w}} \triangleq \max_q \frac{1}{w_q} \sum_{r=1}^{Q} | [\mathbf{A}]_{qr} | w_r, \qquad \mathbf{A} \in \mathbb{R}^{Q \times Q}.
\end{equation}

The mapping $\WFQ(\cdot)$ is said to be a block-contraction\footnote{The mapping $\mathbf{T}$ is called a block-contraction with modulus $\alpha \in [0,1)$ if it is a contraction in the block-maximum norm with modulus $\alpha$ \cite{bk:tsitsiklis1989parallel}.} with modulus $\alpha$ with respect to the norm $\big|\big|\cdot\big|\big|_{2,\block}^{\mathbf{w}}$ if there exists $\alpha \in [0,1)$ such that, $\forall \mathbf{p}^{(1)},\mathbf{p}^{(2)} \in \mathcal{P}$,
\begin{equation}\label{eq:bkCon}
    \big|\big|\WFQ(\mathbf{p}^{(1)})-\WFQ(\mathbf{p}^{(2)})\big|\big|_{2,\block}^{\mathbf{w}} \leq \alpha \big|\big| \mathbf{p}^{(1)} - \mathbf{p}^{(2)} \big|\big|_{2,\block}^{\mathbf{w}},
\end{equation}
where $\mathbf{p}^{(i)}=\Big( \mathbf{p}^{(i)}_q, \dots, \mathbf{p}^{(i)}_q \Big)$ for $i=1,2$.

The contraction property of the waterfilling mapping is given by the following lemma:
%
\begin{lemma}\label{lemma:WFcontract}
Given $\mathbf{w} \triangleq [w_1, \dots, w_Q]^T > \mathbf{0}$, the mapping $\WFQ(\cdot)$ defined in \eqref{eq:blKWF} satisfies
\begin{equation}\label{eq:contractWF}
\begin{split}
    \big|\big|\WFQ\big(&\mathbf{p}^{(1)}\big) - \WFQ\big(\mathbf{p}^{(2)}\big)  \big|\big|_{2,\block}^{\mathbf{w}} \leq \\ &||\mathbf{S}^{max}+\mathbf{E}||_{\infty, \mat}^{\mathbf{w}} \times \big|\big|\mathbf{p}^{(1)}-\mathbf{p}^{(2)}\big|\big|_{2,\block}^{\mathbf{w}},
\end{split}
\end{equation}
$\forall \mathbf{p}^{(1)},\mathbf{p}^{(2)} \in \mathcal{P}$, where $\mathbf{E}$ and $\mathbf{S}$ as defined in \eqref{eq:E_defn} and \eqref{eq:S_max} respectively. Furthermore, if
\begin{equation}\label{eq:lemmaCondn}
    ||\mathbf{S}^{max}+\mathbf{E} ||_{\infty, \mat}^{\mathbf{w}} < 1,
\end{equation}
for some $\mathbf{w}>\mathbf{0}$, then the mapping $\WFQ(\cdot)$ is a block contraction with modulus $\alpha = ||\mathbf{S}^{max}+\mathbf{E} ||_{\infty, \mat}^{\mathbf{w}}$.
\end{lemma}
\begin{proof}
This proof is structured similar to that of \cite[Proposition 2]{jrnl:Opt2}. However the additional term $\epsilon_q \mathbf{f}_q$ in \eqref{eq:WFprj} necessitates a separate treatment here.

For each $q \in \Omega$ and $i=1,2$, given $\mathbf{f}_q^{(i)} \!=\! \left[ \sqrt{\sum_{r \neq q} p_r^2(1)^{(i)}}, \dots, \sqrt{\sum_{r \neq q} p_r^2(N)^{(i)}} \right]^T$, let $\mathbf{p}_{-q}(k)^{(i)} \triangleq [p_1(k)^{(i)}, \ldots, p_{q-1}(k)^{(i)}, p_{q+1}(k)^{(i)}, \ldots, p_Q(k)^{(i)}]$ and $\Delta\mathbf{f}_q \triangleq \big|\big| \mathbf{f}_q^{(1)}- \mathbf{f}_q^{(2)} \big|\big|_2$. Then,
\begin{eqnarray}
\label{eq:f1} \Delta\mathbf{f}_q \!\!\!\!\!& = &\!\!\!\!\! \left[ \sum_{k=1}^N \left( \sqrt{\sum_{r \neq q} p_r^2(k)^{(1)}} - \sqrt{\sum_{r \neq q} p_r^2(k)^{(2)}} \right)^2 \right]^{\frac{1}{2}}  \\
\label{eq:f2} \!\!\!\!\!&=&\!\!\!\!\!
\left[ \sum_{k=1}^N \left( \big|\big| \mathbf{p}_{-q}(k)^{(1)}\big|\big|_2 -\big|\big| \mathbf{p}_{-q}(k)^{(2)}\big|\big|_2 \right)^2 \right]^{\frac{1}{2}}\\
\label{eq:f3} \!\!\!\!\!&\leq&\!\!\!\!\! \left[ \sum_{k=1}^N \Big|\Big| \mathbf{p}_{-q}(k)^{(1)} - \mathbf{p}_{-q}(k)^{(2)} \Big|\Big|_2^2 \right]^{\frac{1}{2}} \\
\!\!\!\!\!&=&\!\!\!\!\!
\Bigg[ \sum_{k=1}^N \sum_{r \neq q} \Big( p_r^2(k)^{(1)} + p_r^2(k)^{(2)} \nonumber\\
\label{eq:f4}   && \qquad \qquad \qquad \qquad - 2 p_r(k)^{(1)} p_r(k)^{(2)} \Big) \Bigg]^{\frac{1}{2}} \\
\label{eq:f5} \!\!\!\!\!&=&\!\!\!\!\! \bigg[ \sum_{r \neq q} \Big|\Big|\mathbf{p}_r^{(1)}-\mathbf{p}_r^{(2)} \Big|\Big|_{2}^2 \bigg]^{\frac{1}{2}}.
\end{eqnarray}
where \eqref{eq:f3} follows from \cite[Lemma 5.1.2]{bk:MatrixAnalysis}. Now, define for each $q \in \Omega$,
\begin{equation}\label{eq:e_wfq1}
\begin{split}
e_{\WFq} &\triangleq \Big|\Big|\WFq\big(\mathbf{p}_{-q}^{(1)}\big)-\WFq\big(\mathbf{p}_{-q}^{(2)}\big)\Big|\Big|_{2} \\
e_q &\triangleq \Big|\Big|\mathbf{p}_q^{(1)}-\mathbf{p}_q^{(2)}\Big|\Big|_{2}.
\end{split}
\end{equation}
Then, using \eqref{eq:WFprj} in \eqref{eq:e_wfq1}, $e_{\WFq}$ can be written as
\begin{eqnarray}
e_{\WFq} \!\!\!\!\!&=&\!\!\!\!\! \bigg|\bigg| \Big[ -\boldsymbol{\sigma}_q - \sum_{r \neq q} \mathbf{F}_{rq} \mathbf{p}_r^{(1)} - \epsilon_q \mathbf{f}_q^{(1)}  \Big]_{\mathcal{P}_q }  \nonumber\\
\label{eq:e_wfq2}  && \qquad   - \Big[ -\boldsymbol{\sigma}_q - \sum_{r \neq q} \mathbf{F}_{rq} \mathbf{p}_r^{(2)} - \epsilon_q \mathbf{f}_q^{(2)}  \Big]_{\mathcal{P}_q}  \bigg|\bigg|_2 \\
\!\!\!\!\!&\leq&\!\!\!\!\! \bigg|\bigg| \sum_{r \neq q} \mathbf{F}_{rq} \mathbf{p}_r^{(1)} + \epsilon_q \mathbf{f}_q^{(1)} \nonumber\\
\label{eq:e_wfq3} &&\qquad\qquad\quad\ -\sum_{r \neq q} \mathbf{F}_{rq} \mathbf{p}_r^{(2)} - \epsilon_q \mathbf{f}_q^{(2)} \bigg|\bigg|_2\\
\label{eq:e_wfq4} \!\!\!\!\!&=&\!\!\!\!\! \bigg|\bigg| \sum_{r \neq q} \mathbf{F}_{rq} \left( \mathbf{p}_r^{(1)} - \mathbf{p}_r^{(2)} \right) + \epsilon_q \left(\mathbf{f}_q^{(1)}- \mathbf{f}_q^{(2)} \right)\!\! \bigg|\bigg|_2\\
\label{eq:e_wfq5} \!\!\!\!\!&\leq&\!\!\!\!\! \bigg|\bigg|\! \sum_{r \neq q} \mathbf{F}_{rq} \left( \mathbf{p}_r^{(1)} - \mathbf{p}_r^{(2)} \right)\!\bigg|\bigg|_2 \!\!+ \epsilon_q \Big|\Big| \mathbf{f}_q^{(1)}- \mathbf{f}_q^{(2)}\! \Big|\Big|_2\\
\!\!\!\!\!&\leq&\!\!\!\!\! \Big|\Big| \sum_{r \neq q} \mathbf{F}_{rq} \left( \mathbf{p}_r^{(1)} - \mathbf{p}_r^{(2)} \right)\Big|\Big|_2 \nonumber\\
\label{eq:e_wfq6} && \qquad\qquad\quad +\ \epsilon_q \bigg[ \sum_{r \neq q} \Big|\Big|\mathbf{p}_r^{(1)}-\mathbf{p}_r^{(2)}\Big|\Big|_{2}^2 \bigg]^{\frac{1}{2}}\\
\!\!\!\!\!&\leq&\!\!\!\!\! \sum_{r \neq q} \Big( \max_k ~ {F}_{rq}(k) \Big) \Big|\Big| \mathbf{p}_r^{(1)} - \mathbf{p}_r^{(2)} \Big|\Big|_2 \nonumber\\
\label{eq:e_wfq7} &&\qquad\qquad\quad + \  \epsilon_q \bigg[ \sum_{r \neq q} \Big|\Big|\mathbf{p}_r^{(1)}-\mathbf{p}_r^{(2)}\Big|\Big|_{2}^2 \bigg]^{\frac{1}{2}}\\
\label{eq:e_wfq8} \!\!\!\!\!&=&\!\!\!\!\! \sum_{r \neq q} \left( \max_{k \in \mathcal{D}_q \cap \mathcal{D}_r} {F}_{rq}(k) \right) e_r + \epsilon_q \Big[ \sum_{r \neq q} \ e_r^2 \Big]^{\frac{1}{2}}\\
\label{eq:e_wfq9} \!\!\!\!\!&\leq&\!\!\!\!\! \sum_{r \neq q} \Big(  [\mathbf{S}^{max}]_{rq} + \epsilon_q \Big) e_r
\end{eqnarray}
$\forall \ \mathbf{p}_q^{(1)}, \mathbf{p}_q^{(2)} \in \mathcal{P}_q$ and $\forall \ q \in \Omega$, where:
\eqref{eq:e_wfq3} follows from the nonexpansive property of the waterfilling projection \cite[Lemma 3]{jrnl:Opt2}; \eqref{eq:e_wfq5} follows from the triangle inequality \cite{bk:MatrixAnalysis}; \eqref{eq:e_wfq6} follows from \eqref{eq:f5}; \eqref{eq:e_wfq7} and \eqref{eq:e_wfq8} follow from the definitions of $\mathbf{F}_{rq}$ and $e_q$ respectively from \eqref{eq:vecSigma} and \eqref{eq:e_wfq1}; and \eqref{eq:e_wfq9} follows from the definition of $\mathbf{S}^{max}$ in \eqref{eq:S_max} and Jensen's inequality \cite{bk:cvxBoyd}.

The set of inequalities in \eqref{eq:e_wfq9} can be written in vector form as
\begin{equation}\label{eq:e_WF1}
    \mathbf{0} \leq \mathbf{e}_{\WFQ} \leq (\mathbf{S}^{max}+\mathbf{E}) \mathbf{e}
\end{equation}
where $\mathbf{E}$ is defined in \eqref{eq:E_defn} and the vectors $\mathbf{e}_{\WFQ}$ and $\mathbf{e}$ are defined as    $\mathbf{e}_{\WFQ} \triangleq \big[ e_{\WFQ_1}, \dots, e_{\WFQ_Q} \big]^T$, and $\mathbf{e} = [e_1. \dots, e_Q]^T$. Using the vector and matrix weighted maximum norms from \eqref{eq:vecNorm} and \eqref{eq:matrixNorm} respectively, \eqref{eq:e_WF1} can be written as
\begin{equation}\label{eq:e_WF2}
\begin{split}
    \big|\big| \mathbf{e}_{\WFQ} \big|\big|_{\infty,\VEC}^{\mathbf{w}} &\leq \big|\big| (\mathbf{S}^{max}+\mathbf{E}) \mathbf{e} \big|\big|_{\infty,\VEC}^{\mathbf{w}} \\
    &\leq \big|\big| \mathbf{S}^{max}+\mathbf{E} \big|\big|_{\infty,\mat}^{\mathbf{w}}  \cdot \big|\big|\mathbf{e} \big|\big|_{\infty,\VEC}^{\mathbf{w}},
\end{split}
\end{equation}
$\forall \ \mathbf{w}>\mathbf{0}$. Using the block-maximum norm \eqref{eq:WFbknorm}, we get
\begin{equation}\label{eq:e_WF3}
\begin{split}
    \big|\big|\WFQ\big(\mathbf{p}^{(1)}\big) &- \WFQ\big(\mathbf{p}^{(2)}\big)\big|\big|_{2,\block}^{\mathbf{w}} = \big|\big|\mathbf{e}_{\WFQ} \big|\big|_{\infty,\VEC}^{\mathbf{w}} \\
    \leq& \big|\big| \mathbf{S}^{max}+\mathbf{E} \big|\big|_{\infty,\mat}^{\mathbf{w}} \Big|\Big|\mathbf{p}_r^{(1)}-\mathbf{p}_r^{(2)}\Big|\Big|_{2,\block}^{\mathbf{w}}
\end{split}
\end{equation}
$\forall \mathbf{p}^{(2)},\mathbf{p}^{(2)} \in \mathcal{P}$, with $\mathbf{E}$ and $\mathbf{S}$ as defined in \eqref{eq:E_defn} and \eqref{eq:S_max} respectively. It is clear that $\WFQ(\cdot)$ is a block contraction when $||\mathbf{S}^{max}+\mathbf{E}||_{\infty,\mat}^{\mathbf{w}} < 1$.
\end{proof}
\subsection{Proof of Theorem~\ref{thrm:B_existNE}}\label{appendix:B_existNE}
From \cite{Rosen1965}, every concave game\footnote{A game is said to be concave if the payoff functions are concave and the sets of admissible strategies are compact and convex.} has at least one equilibrium. For the game $\mathscr{G}^{\rob}$:
\begin{enumerate}
  \item The set of feasible strategy profiles, $\mathcal{P}_q$ of each player $q$ is compact and convex.
  \item The payoff function of each player $q$ in \eqref{eq:WCp2} is continuous in $\mathbf{p} \in \mathcal{P}$ and concave in $\mathbf{p}_q \in \mathcal{P}_q$.
\end{enumerate}
Thus, the game $\mathscr{G}^{\rob}$ has at least one robust equilibrium.

From Lemma~\ref{lemma:WFcontract}, the waterfilling mapping $\WFQ(\cdot)$ is a block-contraction if \eqref{eq:lemmaCondn} is satisfied for some $\mathbf{w>0}$. Thus, the RE of game $\mathscr{G}^{\rob}$ is unique (using \cite[Theorem 1]{jrnl:JSAC}). Since $\mathbf{S}^{max}+\mathbf{E}$ is a nonnegative matrix, there exists a positive vector $\bar{\mathbf{w}}$ such that
\begin{equation}\label{eq:NEproof}
    ||\mathbf{S}^{max}+\mathbf{E}||_{\infty,\mat}^{\bar{\mathbf{w}}} < 1
\end{equation}
Using \cite[Corollary 6.1]{bk:tsitsiklis1989parallel} and the triangle inequality \cite{bk:MatrixAnalysis}, this is satisfied when
\begin{equation}\label{eq:NE_proof2}
    ||\mathbf{S}^{max}||_{\infty,\mat}^{\bar{\mathbf{w}}}+||\mathbf{E}||_{\infty,\mat}^{\bar{\mathbf{w}}} < 1 \; \Rightarrow \; \rho(\mathbf{S}^{max}) < 1- \rho(\mathbf{E}).
\end{equation}

From Lemma~\ref{lemma:WFcontract} and \eqref{eq:NEproof} the waterfilling mapping $\WFQ(\cdot)$ is a block-contraction. From \cite[Theorem 2]{jrnl:JSAC}, the robust iterative waterfilling algorithm described in Algorithm~\ref{al:RIWFA} converges to the unique RE of game $\mathscr{G}^{\rob}$ for any set of feasible initial conditions and any update schedule.
%
%
\subsection{Proof of Theorem~\ref{thrm:2F}}\label{appendix:2F}
Consider the interior operating points of the robust waterfilling operator $\WFq(\cdot)$ where it is linear. Eliminating $\mu$ from \eqref{eq:2pow_allocs1}, we get
\begin{equation}\label{eq:2p}
   p = \frac{1 - \alpha - \epsilon}{2(1 - (m+1)\alpha/2 - \epsilon)} \geq 0.5.
\end{equation}

The signal-to-interference-plus-noise ratio (SINR) for the two users in the two frequency bins is given by
\begin{equation}\label{eq:2sinr}
\begin{split}
   \SINr(1) &= \SINR(2) = \frac{p}{\sigma^2 + \alpha(1-p)} \\
   \SINr(2) &= \SINR(1) = \frac{1-p}{\sigma^2 + m \alpha p} \ .
\end{split}
\end{equation}
and the sum-rate of the system at the RE is
\begin{equation}\label{eq:2sum_rate}
    S_{\rob} \triangleq 2 \log \Big(1+ \frac{p}{\sigma^2 + \alpha(1-p)} \Big) + 2 \log \Big(1+ \frac{1-p}{\sigma^2 + m \alpha p} \Big)
\end{equation}

The gradient of $p$ with respect to $\epsilon$ is
\begin{equation}\label{eq:2dp_de}
    \frac{\der p}{\der \epsilon} = \frac{(m-1)\alpha}{4(1 - (m+1)\alpha/2 - \epsilon)^2} >0.
\end{equation}
Thus, the RE moves towards the FDMA solution as the uncertainty bound increases.
\subsubsection*{Case 1: High interference scenario}
In the high interference scenario, $\sigma^2 \ll \alpha(1-p)$. Let $\xi={p}/\alpha(1-p)$. Then, the SINR for the two users in the two frequency bins can be approximated as
\begin{equation}\label{eq:2sinrHI}
\begin{split}
   \SINr(1) &= \SINR(2) \approx \frac{p}{\alpha(1-p)} = \xi \ ,  \\
   \SINr(2) &= \SINR(1) \approx \frac{1-p}{m \alpha p} = \frac{1}{m \alpha^2 \xi}.
\end{split}
\end{equation}
The sum-rate of the system at high interference can be approximated as
\begin{equation}\label{eq:2sum_rateHI}
\begin{split}
    S &\approx 2 \log \big( (1+\xi)(1+\frac{1}{m \alpha^2 \xi}) \big)\\
      &= 2 \log \big( 1+\frac{1}{m \alpha^2 } +\xi+\frac{1}{m \alpha^2 \xi} \big).
\end{split}
\end{equation}

Our aim is to analyse the behaviour of the sum-rate $S$ as the uncertainty $\epsilon$ increases. To this end, we show that the gradient of the sum-rate with respect to $\epsilon$ is positive. As $\log(x)$ increases monotonically with $x$, we consider
\begin{equation}\label{eq:2anlsysHI1}
    \frac{\der}{\der \epsilon} (\xi+\frac{1}{m \alpha^2 \xi})= (1 - \frac{1}{m \alpha^2 \xi^2})\frac{\der \xi}{\der \epsilon} 
    = \Big(1- \frac{(1-p)^2}{m p^2}\Big)\frac{\der \xi}{\der \epsilon},
\end{equation}
and $(1- \frac{(1-p)^2}{m p^2})>0$ since $p\geq0.5$ and $m>1$. Now,
\begin{equation}\label{eq:2anlsysHI2}
    \frac{\der \xi}{\der \epsilon} = \frac{1}{\alpha (1-p)^2} \frac{\der p}{\der \epsilon} .
\end{equation}
From \eqref{eq:2dp_de}, \eqref{eq:2anlsysHI1} and \eqref{eq:2anlsysHI2}, we get $\frac{\der S}{\der \epsilon} > 0$. Thus, we see that the sum-rate of the system increases as the uncertainty $\epsilon$ increases. This also shows that the robust-optimization equilibrium achieves a higher sum-rate in the presence of channel uncertainty $(\epsilon>0)$ than the Nash equilibrium at zero uncertainty $(\epsilon=0)$.

The social optimal solution for this system at high interference is frequency division multiplexing \cite{jrnl:ZQLuo}. In other words, the frequency space is fully partitioned at the social optimal solution. The sum-rate at the social optimal solution for the given system at high interference, $S^*$, is given by
\begin{equation}\label{eq:2SOsum_rateHI}
    S^* = 2 \log \Big( 1+\frac{1}{\sigma^2} \Big).
\end{equation}
The price of anarchy at high interference, $\PoA$, is
\begin{equation}\label{eq:2PoAHI}
   \PoA = \frac{ \log \big( 1+\frac{1}{\sigma^2} \big)}{\log \big( 1+\frac{1}{m \alpha^2 } +\xi+\frac{1}{m \alpha^2 \xi} \big)}.
\end{equation}
Since $\frac{\der S}{\der \epsilon} >0$, we have $\frac{\der \PoA}{\der \epsilon~~~} < 0$.
\subsubsection*{Case 2: Low interference scenario}
In the low interference scenario, i.e. when $m \alpha p \ll \sigma^2$, the signal-to-interference+noise ratio SINR for the two users in the two frequency bins can be approximately written as
\begin{equation}\label{eq:2sinrHI}
\begin{split}
   \SINr(1) &= \SINR(2) \approx \frac{p}{\sigma^2}, \\
   \SINr(2) &= \SINR(1) \approx \frac{1-p}{\sigma^2}.
\end{split}
\end{equation}
The sum-rate of the system at low interference can be approximated as
\begin{equation}\label{eq:2sum_rateLI}
\begin{split}
    S &\approx 2 \log \big( (1+\frac{p}{\sigma^2})(1+\frac{1-p}{\sigma^2}) \big)\\
      &= 2 \log \big( 1+\frac{1}{\sigma^2}+\frac{p - p^2}{\sigma^2} \big).
\end{split}
\end{equation}
Now,
\begin{equation}\label{eq:2anlsysLI_1}
    \frac{\der S}{\der \epsilon} = \Big( 1+\frac{1}{\sigma^2}+\frac{p - p^2}{\sigma^2} \Big) ^{-1} \ \frac{(1-2p)}{\sigma^2} \frac{\der p}{\der \epsilon} <0.
\end{equation}

At low interference, the system behaves similar to a parallel Gaussian channel system. The social optimal solution in this scenario is the classical waterfilling solution and leads to equal power allocation, i.e., $p_1(1)=p_1(2)=p_2(1)=p_2(2)=p=0.5$. The sum-rate at the social optimal solution for the given system at low interference, $S^*$, is given by
\begin{equation}\label{eq:2SOsum_rateHI}
    S^* = 4 \log \Big( 1+\frac{1}{2\sigma^2} \Big).
\end{equation}
The price of anarchy at low interference, $\PoA$, is
\begin{equation}\label{eq:2PoAHI}
\begin{split}
   \PoA &= \frac{ 4 \log \big( 1+\frac{1}{2 \sigma^2} \big)}{2\log \big( 1+\frac{1}{\sigma^2}+\frac{p - p^2}{\sigma^2} \big)} \\
   &= \frac{ \log \big( 1+\frac{1}{\sigma^2}+\frac{1}{4\sigma^4} \big)}{\log \big( 1+\frac{1}{\sigma^2}+\frac{p - p^2}{\sigma^2} \big)}.
\end{split}\end{equation}

Note that, at low interference, $m \alpha p \ll 1$. From \eqref{eq:2p}, we get $p \approx 0.5$. Thus the $\PoA$ is close to unity. Since $\frac{\der S}{\der \epsilon} <0$, we have $\frac{\der \PoA }{\der \epsilon~~~} > 0 $.
\qed
\subsection{Proof of Proposition~\ref{prop:crit}}\label{crit_proof}
The gradient of the sum-rate $S_{\rob}$ with respect to $\epsilon$ is
\begin{equation}\label{eq:dSde}
    \frac{\der S_{rob}}{\der \epsilon~~~{}} = \frac{\der S_{\rob}}{\der p} ~\frac{\der p}{\der \epsilon}.
\end{equation}
From \eqref{eq:2dp_de}, we have $\frac{\der p}{\der \epsilon} > 0$. Now,
\begin{equation}\label{eq:2dsdp_c}
\begin{split}
    \frac{\der S_{\rob}}{\der p~~~{}} = 2~ &\frac{ \frac{1}{\sigma^2+\alpha(1-p)} + \frac{\alpha p }{ \left(\sigma^2+\alpha(1-p)\right)^2}}{1+ \frac{p}{\sigma^2+\alpha(1-p)}} \\
    &\qquad\qquad - 2~ \frac{ \frac{1}{ \sigma^2+m\alpha p} + \frac{ (1-p)m\alpha}{\left( \sigma^2+m\alpha p \right)^2}}{1+ \frac{1-p}{ \sigma^2+m\alpha p}}.
\end{split}
\end{equation}
Setting $\frac{\der S_{\rob}}{\der p~~~{}} = 0$, we solve for $\alpha$ to get the following roots,
\begin{equation}\label{eq:alpha_roots}
    \alpha =
    \begin{cases}
      0,\\
      \frac{-\sigma^2}{2m}\left( m + 1 \pm \Big( 4m/\sigma^2 + (m+1)^2 \Big)^{\frac{1}{2}} \right),\\
      \frac{ \sigma^2 (2p - 1)}{ (m-1)p^2 + 2p -1 }.
    \end{cases}
\end{equation}
The positive root that is independent of $\epsilon$ and $p$ (which is a function of the uncertainty $\epsilon$, from \eqref{eq:2p}) is the required solution where the sum-rate is constant regardless of uncertainty. Thus, the required interference value is given by
\begin{equation}\label{eq:aplha_crit_INproof}
    \alpha_{\crit} = \frac{\sigma^2}{2m}\left( \left( 4m/\sigma^2 + (m+1)^2 \right)^{\frac{1}{2}} - m - 1 \right).
\end{equation}
Since the root $\alpha_{\crit}$ of $\frac{\der S_{\rob}}{\der p~~~{}}$ is independent of $p$, different power allocation schemes (resulting in different values of $p$) will result in the same sum-rate at $\alpha=\alpha_{\crit}$. Thus, the price of anarchy at $\alpha=\alpha_{\crit}$ is unity.
\qed
\subsection{Proof of Lemma~\ref{lemma:degPart}}\label{appendix:degPart}
From \eqref{eq:WF}, the power allocations for the two users at the robust-optimization equilibrium in the $k$-th frequency are
\begin{equation}\label{eq:pow_alloc}
\begin{split}
    p_1(k)&= \Big(\mu_1 - \sigma^2 - (F_{21}+\epsilon)p_2(k)\Big)^+, \\
    p_2(k)&= \Big(\mu_2 - \sigma^2 - (F_{12}+\epsilon)p_1(k)\Big)^+,
\end{split}
\end{equation}
with $\sum_{k=1}^{N} p_1(k)=\sum_{k=1}^{N} p_2(k)=P_T$.

Let $\mathcal{D}_1, \mathcal{D}_2$ and $\mathcal{D}_{ol}$ be the sets of frequencies exclusively used by user 1, user 2 and by both respectively and $n_1 \triangleq |\mathcal{D}_1|$ and $n_2 \triangleq |\mathcal{D}_2|$ be the number of frequencies exclusively used by user 1 and user 2 respectively at the equilibrium. Then, from \eqref{eq:pow_alloc}, we have $p_1(k) = \mu_1-\sigma^2_1$ and $p_2(k)=0 \ \forall \ k \in \mathcal{D}_1$ and $p_1(k)=0$ and $p_2(k) = \mu_2-\sigma^2_1 \ \forall \ k \in \mathcal{D}_2$. The power remaining for allocation to the frequencies in $\mathcal{D}_{ol} \triangleq \{k_1, \ldots, k_{ol}\}$ by user 1 and user 2 is $(1-n_1\mu_1)$ and $(1-n_2\mu_2)$ respectively.

This separation of the frequency-space into exclusive-use and overlapped-use frequencies allows us to analyse the system without the nonlinear operation $(\cdot)^+$. Thus, we can write the power allocations at the fixed point in the overlapped-use frequency-space as a system of linear equations,
\begin{eqnarray}
\label{eq:lin_sys1}
    p_1(k) + \big( F_{21}(k) + \epsilon \big)p_2(k) - \mu_1 - \sigma^2 \!&=&\! 0, \ \ k \in \mathcal{D}_{ol} \ \quad \\
\label{eq:lin_sys2}
    \big( F_{12}(k) + \epsilon \big)p_1(k) + p_2(k) - \mu_2 - \sigma^2 \!&=&\! 0, \ \ k \in \mathcal{D}_{ol} \\
\label{eq:lin_sys3}
    \textstyle\sum_{k \in \mathcal{D}_{ol}} p_1(k)+n_1(\mu_1 - \sigma^2) &=&\! P_T, \\
    \textstyle\sum_{k \in \mathcal{D}_{ol}} p_2(k)+n_2(\mu_2 - \sigma^2) &=&\! P_T.
\end{eqnarray}
Writing these in matrix form we get,
\begin{equation}\label{eq:matrixEqn}
    \begin{bmatrix}
        \mathbf{A}_{k_1} &    {}    & \mathbf{0} & -\mathbf{I}_2\\
            {}         & \ddots   & {} & \vdots \\
        \mathbf{0}     &    {}    & \mathbf{A}_{k_{ol}} & -\mathbf{I}_2\\
        \mathbf{I}_2     &    \cdots    & \mathbf{I}_2 & \mathbf{D} \end{bmatrix}
        \begin{bmatrix} \mathbf{p}(k_1) \\ \vdots  \\ \mathbf{p}(k_{ol}) \\ \boldsymbol{\mu} \end{bmatrix} =
        \begin{bmatrix} \mathbf{0}_2 \\ \vdots  \\ \mathbf{0}_2 \\ \mathbf{p_t} \end{bmatrix}
\end{equation}
where
\begin{equation}\label{eq:defn_matrixEqn1}
\begin{aligned}
    \mathbf{A}_k &\triangleq \begin{bmatrix} 1 & F_{21}(k)+\epsilon\\ F_{12}(k)+\epsilon & 1 \end{bmatrix}, \
     \mathbf{D} \triangleq \begin{bmatrix} n_1 & 0 \\ 0 & n_2 \end{bmatrix}, \\
     \mathbf{p}(k) &\triangleq \begin{bmatrix} p_1(k) \\ p_2(k) \end{bmatrix}, \quad
     \mathbf{p_t} \triangleq \begin{bmatrix} P_T \\ P_T \end{bmatrix} \quad
     \boldsymbol{\mu} \triangleq \begin{bmatrix} \mu_1 - \sigma^2 \\ \mu_2 - \sigma^2 \end{bmatrix}.
\end{aligned}
\end{equation}
Let
\begin{equation}\label{eq:defn_matrixEqn2}
\begin{aligned}
    \mathbf{A} &\triangleq \begin{bmatrix} \mathbf{A}_{k_1} & {} & \mathbf{0}\\ {} & \ddots &{}\\ \mathbf{0} & {}& \mathbf{A}_{k_{ol}} \end{bmatrix}, \
    \mathbf{B} \triangleq \begin{bmatrix} -\mathbf{I}_2 \\ \vdots  \\ -\mathbf{I}_2 \end{bmatrix}, \\
    \mathbf{C} &\triangleq \big[\mathbf{I}_2 \ \dots \ \mathbf{I}_2\big] \ \text{and} \
    \mathbf{P} \triangleq \big[ \mathbf{p}(k_1)  \dots   \mathbf{p}(k_{ol}) \big]^T.
\end{aligned}
\end{equation}
so that we can write \eqref{eq:matrixEqn} as
\begin{equation}\label{b1}
    \begin{bmatrix} \mathbf{A} &  \mathbf{B}\\ \mathbf{C} & \mathbf{D} \end{bmatrix}
        \begin{bmatrix} \mathbf{P} \\ \boldsymbol{\mu} \end{bmatrix} =
        \begin{bmatrix} \mathbf{0}\\ \mathbf{p_t} \end{bmatrix}.
\end{equation}
We can solve this system to get
\begin{eqnarray}
        \begin{bmatrix} \mathbf{P} \\ \boldsymbol{\mu} \end{bmatrix} =
        \begin{bmatrix} \mathbf{A} &  \mathbf{B}\\ \mathbf{C} & \mathbf{D} \end{bmatrix}^{-1}
        \begin{bmatrix} \mathbf{0}\\ \mathbf{p_t} \end{bmatrix} 
\label{eq:P_solv2}
        =\begin{bmatrix} \mathbf{W} &  \mathbf{X}\\ \mathbf{Y} & \mathbf{Z} \end{bmatrix}
        \begin{bmatrix} \mathbf{0}\\   \mathbf{p_t} \end{bmatrix} = \begin{bmatrix} \mathbf{Xp_t}\\ \mathbf{Zp_t} \end{bmatrix},
\end{eqnarray}
where
\begin{equation}\label{eq:ABCD_inv}
    \begin{bmatrix}
        \mathbf{W} &  \mathbf{X}\\
        \mathbf{Y} & \mathbf{Z} \end{bmatrix} \triangleq
        \begin{bmatrix}
        \mathbf{A} &  \mathbf{B}\\
        \mathbf{C} & \mathbf{D} \end{bmatrix}^{-1}.
\end{equation}
Using \cite[Fact 10.12.9]{bk:bernstein2009matrix} and differentiating \eqref{eq:P_solv2} with respect to $\epsilon$, we get
\begin{eqnarray}
\label{eq:dPdE1}
       \D\begin{bmatrix}
       \mathbf{ P} \\
       \boldsymbol{\mu} \end{bmatrix} \!\!\!\!&=& \!\!
       -\begin{bmatrix} \mathbf{A} &  \mathbf{B}\\ \mathbf{C} & \mathbf{D} \end{bmatrix}^{-1} \!
       \begin{bmatrix} \D \mathbf{A} &  \mathbf{0}\\ \mathbf{0} & \mathbf{0} \end{bmatrix}
       \begin{bmatrix} \mathbf{A} &  \mathbf{B}\\ \mathbf{C} & \mathbf{D} \end{bmatrix}^{-1} \!
       \begin{bmatrix} \mathbf{0}\\ \mathbf{p_t} \end{bmatrix}  \  \quad \\
\label{eq:dPdE3}
       &=&\!\! -\begin{bmatrix} \mathbf{W}(\D\mathbf{A}) \mathbf{X} \mathbf{p_t}\\
       \mathbf{Y}(\D\mathbf{A}) \mathbf{X} \mathbf{p_t} \end{bmatrix} .
\end{eqnarray}
Due to the nature of the waterfilling function, $n_1$ and $n_2$ are non-decreasing piecewise-constant functions of $\epsilon$. The above derivative exists only in the regions where $n_1$ and $n_2$ are constant. From \cite[Proposition 2.8.7]{bk:bernstein2009matrix} and using $\mathbf{A}^{-1} = \Diag ( \mathbf{A}_{k_1}^{-1}, \dots, \mathbf{A}_{k_{ol}}^{-1})$,
we get
\begin{eqnarray}
\label{eq:W2}
    \mathbf{W} \!\!\!\!&=&\!\!\!\! \begin{bmatrix} \mathbf{A}_{k_1}^{-1}-\mathbf{A}_{k_1}^{-1}\mathbf{Z}\mathbf{A}_{k_1}^{-1} & \cdots & -\mathbf{A}_{k_1}^{-1}\mathbf{Z}\mathbf{A}_{k_{ol}}^{-1} \\ \vdots & {} & \vdots\\ -\mathbf{A}_{k_{ol}}^{-1}\mathbf{Z}  \mathbf{A}_{k_1}^{-1}  & \cdots& \mathbf{A}_{k_{ol}}^{-1}-\mathbf{A}_{k_{ol}}^{-1}\mathbf{Z}\mathbf{A}_{k_{ol}}^{-1} \end{bmatrix}, ~ \quad\\
\label{eq:X2}
    \mathbf{X} \!\!\!\!&=&\!\!\!\! \begin{bmatrix} \mathbf{A}_{k_1}^{-1}\mathbf{Z} & \dots  & \mathbf{A}_{k_{ol}}^{-1}\mathbf{Z} \end{bmatrix}^T, \\
\label{eq:Y2}
    \mathbf{Y} \!\!\!\!&=&\!\!\!\! -\begin{bmatrix} \mathbf{Z}\mathbf{A}_{k_1}^{-1} \ \cdots  \ \mathbf{Z}\mathbf{A}_{k_{ol}}^{-1} \end{bmatrix}, \\
\label{eq:Z2}
    \mathbf{Z} \!\!\!\!&=&\!\!\!\! \frac{1}{\widehat{\Delta}}\begin{bmatrix} n_2+\displaystyle\sum_{k \in \mathcal{D}_{ol}}\frac{1}{\Delta_i} & \displaystyle\sum_{k \in \mathcal{D}_{ol}} \frac{F_{21}(i)+\epsilon}{\Delta_i}\\ \displaystyle\sum_{k \in \mathcal{D}_{ol}}\frac{F_{12}(i)+\epsilon}{{\Delta_i}} & n_1+\displaystyle\sum_{k \in \mathcal{D}_{ol}}\frac{1}{\Delta_i} \end{bmatrix},
\end{eqnarray}
where $\Delta_i \triangleq \det(\mathbf{A}_i) = 1 - \big(F_{21}(i)+\epsilon\big)\big(F_{12}(i)+\epsilon\big)$,
\begin{equation}\label{eq:Delta_hat}
\begin{split}
    \widehat{\Delta} \triangleq& \left(n_1 +  \textstyle\sum_{k \in \mathcal{D}_{ol}}\frac{1}{\Delta_i}\right) \left(n_2+\textstyle\sum_{k \in \mathcal{D}_{ol}}\frac{1}{\Delta_i}\right)\\
    &\qquad - \left(\textstyle\sum_{k \in \mathcal{D}_{ol}} \frac{F_{21}(i)+\epsilon}{\Delta_i}\right) \left(\textstyle\sum_{k \in \mathcal{D}_{ol}} \frac{F_{12}(i)+\epsilon}{\Delta_i} \right),
\end{split}
\end{equation}
and
\begin{equation}\label{eq:A_i_inv}
    \mathbf{A}_{i}^{-1} = \begin{bmatrix} \frac{1}{\Delta_i} & -\frac{F_{21}(i)+\epsilon}{\Delta_i}\\ -\frac{F_{12}(i)+\epsilon}{\Delta_i} & \frac{1}{\Delta_i} \end{bmatrix}.
\end{equation}
Thus, from \eqref{eq:P_solv2} and \eqref{eq:dPdE3}, we get
\begin{equation}\label{eq:P_solved}
    \mathbf{P} = \mathbf{Xp_t} =
    \begin{bmatrix} \mathbf{A}_{k_1}^{-1}\mathbf{Zp_t} & \dots  & \mathbf{A}_{k_{ol}}^{-1}\mathbf{Zp_t} \end{bmatrix}^T,
\end{equation}
and
\begin{eqnarray}
  \frac{\der \mathbf{P}}{\der \epsilon} \!\!\!\!&=&\!\!\! -\mathbf{W}(\frac{\der \mathbf{A}}{\der \epsilon}) \mathbf{X},   \\
\label{eq:dPde_solved}
    &=&\!\!\!\! \begin{bmatrix} \! \displaystyle\sum_{i=k_1}^{k_{ol}}\! \mathbf{A}_{k_1}^{-1}\mathbf{Z}\mathbf{A}_{i}^{-1}\mathbf{GA}_{i}^{-1 }\mathbf{Z}\mathbf{p_t} - \! \mathbf{A}_{k_1}^{-1}\mathbf{GA}_{k_1}^{-1}\mathbf{Z}\mathbf{p_t} \! \\
    \vdots  \\
    \! \displaystyle\sum_{i=k_1}^{k_{ol}} \! \mathbf{A}_{k_{ol}}^{-1}\mathbf{Z}\mathbf{A}_{i}^{-1}\mathbf{GA}_{i}^{-1}\mathbf{Z} \mathbf{p_t} -\! \mathbf{A}_{k_{ol}}^{-1}\mathbf{GA}_{k_{ol}}^{-1}\mathbf{Z}\mathbf{p_t} \!\end{bmatrix}~~~~~
\end{eqnarray}
where
\begin{equation}\label{eq:G}
    \mathbf{G}=\frac{\der \mathbf{A}_i}{\der \epsilon}=\begin{bmatrix} 0 & 1\\ 1 & 0 \end{bmatrix} \ \forall \ i=1,\ldots,N.
\end{equation}
Therefore, for $k=k_1,\dots,k_{ol}$,
\begin{eqnarray}
\label{eq:p(k)_solved}
    \mathbf{p}(k) \!\!\!&=&\!\!\! \begin{bmatrix} p_1(k) & p_2(k) \end{bmatrix}^T = \mathbf{A}_k^{-1}\mathbf{Zp_t}, \\
   \mathbf{p}'(k)\!\!\!&=&\!\!\!  \D \mathbf{p}(k) =\begin{bmatrix} p'_1(k) & p'_2(k) \end{bmatrix}^T \\
  &=&\!\!\! \mathbf{A}_{k}^{-1}\mathbf{Z} \! \displaystyle\sum_{i=k_1}^{k_{ol}}\! \mathbf{A}_{i}^{-1}\mathbf{GA}_{i}^{-1} \mathbf{Z} \mathbf{p_t} \nonumber\\
\label{eq:dp(k)_solved} &&\qquad\qquad\qquad   -\mathbf{A}_{k}^{-1}\mathbf{GA}_{k}^{-1}\mathbf{Z}\mathbf{p_t}
\end{eqnarray}

Consider the extent of partitioning $J(k)$ at any frequency $k$. For frequencies where $J(k)=0$, at least one of the users has zero power allocation and that will not change with change in uncertainty. Thus $\D J(k) = 0$ for $k \in \mathcal{D}_1 \cup \mathcal{D}_2$. Also, in cases when $n_1$ or $n_2$ change due to some frequency $\bar{k}$ dropping from the set $\mathcal{D}_{ol}$, $J(\bar{k})$ increases from some negative value to zero.

Now consider the extent of partitioning for frequencies where both users have non-zero power allocation. Differentiating the extent of partitioning for frequency $k \in \mathcal{D}_{ol}$ with respect to $\epsilon$, we get
\begin{eqnarray}
  \D J(k) &=& \D \left(-p_1(k)p_2(k)\right)\\
   &=& -p'_1(k)p_2(k)+p_1(k)p'_2(k)  \\
  &=& -\mathbf{p}(k)^T \mathbf{G} \mathbf{p}'(k)\\
  &=& -\Big(\mathbf{A}_k^{-1}\mathbf{Zp_t} \Big)^T \mathbf{G} \Big( -\mathbf{A}_{k}^{-1} \mathbf{GA}_{k}^{-1} \mathbf{Z} \mathbf{p_t} \nonumber\\
  && \qquad\  + \mathbf{A}_{k}^{-1}\mathbf{Z} \sum_{i=k_1}^{k_{ol}} \mathbf{A}_{i}^{-1}\mathbf{GA}_{i}^{-1}\mathbf{Z} \mathbf{p_t} \Big) \\
  &=& \mathbf{p_t}^T \mathbf{Z}^T {\mathbf{A}_k^{-1}}^T \mathbf{G} \mathbf{A}_{k}^{-1} \Big( \mathbf{I} \nonumber\\
  && - \mathbf{Z} \! \sum_{i=k_1}^{k_{ol}} \! \mathbf{A}_{i}^{-1}\mathbf{GA}_{i}^{-1}\mathbf{G}^{-1} \mathbf{A}_{k} \Big) \mathbf{GA}_{k}^{-1} \mathbf{Z} \mathbf{p_t} \nonumber\\[-1em]
\end{eqnarray}
Let $\mathbf{q}_k \triangleq \mathbf{GA}_{k}^{-1} \mathbf{Z} \mathbf{p_t}$. Using $\mathbf{G}^T=\mathbf{G}^{-1} = \mathbf{G}$, $\mathbf{GA}_{i}^{-1}\mathbf{G}= \mathbf{A}_i^{-T}$ and $ \mathbf{GA}_{k}\mathbf{G} = \mathbf{A}_k^{T}$ we get
\begin{equation}
  \D J(k) = \mathbf{q}_k^T \Big( \mathbf{A}_{k}^{-1} - \mathbf{A}_{k}^{-1} \mathbf{Z} \sum_{i=k_1}^{k_{ol}} \mathbf{A}_{i}^{-1}\mathbf{A}_{i}^{-T}\mathbf{A}_{k}^T  \Big) \mathbf{q}_k
\end{equation}
Let $\mathbf{M}_k = \sum_{i=k_1}^{k_{ol}}\mathbf{A}_{i}^{-1}\mathbf{A}_{i}^{-T}\mathbf{A}_{k}^T$ and $\mathbf{Q}_k =  \mathbf{A}_{k}^{-1} - \mathbf{A}_{k}^{-1} \mathbf{Z} \mathbf{M}_k $. When $n_{ol} = o(N)$ (i.e., when $\lim_{N \rightarrow \infty} \frac{n_{ol}}{N} = 0$), we have the total number of frequencies, $n_1+n_2 = O(N)$. Since $\mathbf{A}_i^{-1} \mathbf{A}_i^{-T} \mathbf{A}^T_k = O(1)$ for each $i$ and $k$, we have $\mathbf{M}_k = O(n_{ol})$ and $\mathbf{Z} = O(1/N)$. Thus,
\begin{equation}\label{eqn}
\begin{aligned}
   \lim_{N \rightarrow \infty} \mathbf{A}_{k}^{-1} \mathbf{ZM}_k &= \mathbf{0} \\
   \Rightarrow \ \lim_{N \rightarrow \infty} \mathbf{Q}_k + \mathbf{Q}_k^T &= \mathbf{A}_{k}^{-1}+\mathbf{A}_{k}^{-T} \succ \mathbf{0}
\end{aligned}
\end{equation}
from Theorem~\ref{thrm:B_existNE}. Thus we get $\mathbf{x}^T\mathbf{Q}_k\mathbf{x} > 0 \ \forall \mathbf{x} \in \mathbb{R}^{2 \times 1}$ as its symmetric part $\mathbf{Q}_k + \mathbf{Q}_k^T$ is positive definite \cite{johnson1970positive}. Hence, we get $\D J(k) \geq 0$ when $N \rightarrow \infty$, with equality when $J(k)=0$.
\qed
\section*{Acknowledgements}
We thank Dr. Ishai Menache, Microsoft Research for his input on robust game theory and Dr. Gesualdo Scutari, University of Illinois at Urbana-Champaign for initial guidance and advice on waterfilling algorithms. We also thank Peter von Wrycza, KTH Royal Institute of Technology and Dr. M.R Bhavani Shankar, University of Luxembourg for pointing out a typographical error in an early version of the proofs, along with Prof. Bj\"{o}rn Ottersten, KTH Royal Institute of Technology for valuable discussions. We also thank the anonymous reviewers for their valuable feedback.
%
%
\bibliographystyle{IEEEtr}
\bibliography{IEEEabrv,TVT_refs}

\begin{thebibliography}{10}

\bibitem{conf:Amod_Icassp10}
{A.~J.~G. Anandkumar, A. Anandkumar, S. Lambotharan and J. Chambers}, ``Robust
  rate-maximization game under bounded channel uncertainty,'' {\em 2010 IEEE
  International Conference on Acoustics, Speech, and Signal Processing
  (ICASSP)}, pp.~3158 --3161, Mar. 2010.

\bibitem{conf:Amod_Asilomar10}
{A.~J.~G. Anandkumar, A. Anandkumar, S. Lambotharan and J. Chambers},
  ``Efficiency of rate-maximization game under bounded channel uncertainty,''
  {\em 2010 Conference Record of the Forty-Third Asilomar Conference on
  Signals, Systems and Computers}, Nov. 2010.

\bibitem{bk:mackenzie2006game}
A.~B. MacKenzie and L.~A. DaSilva, ``Game theory for wireless engineers,'' {\em
  Synthesis Lectures on Communications}, vol.~1, no.~1, pp.~1--86, 2006.

\bibitem{jrnl:Zehavi_CoopGT}
A.~Leshem and E.~Zehavi, ``Cooperative game theory and the {Gaussian}
  interference channel,'' {\em {IEEE} J. Sel. Areas Commun.}, vol.~26, pp.~1078
  --1088, Sep. 2008.

\bibitem{jrnl:Berry_JSAC06}
J.~Huang, R.~Berry, and M.~Honig, ``Distributed interference compensation for
  wireless networks,'' {\em {IEEE} J. Sel. Areas Commun.}, vol.~24, pp.~1074 --
  1084, May 2006.

\bibitem{mag:GT_FF_GIC}
E.~Larsson, E.~Jorswieck, J.~Lindblom, and R.~Mochaourab, ``Game theory and the
  flat-fading {Gaussian} interference channel,'' {\em {IEEE} Signal Process.
  Mag.}, vol.~26, pp.~18 --27, Sep. 2009.

\bibitem{mag:GT_FS_GIC}
A.~Leshem and E.~Zehavi, ``Game theory and the frequency selective interference
  channel,'' {\em {IEEE} Signal Process. Mag.}, vol.~26, pp.~28 --40, Sep.
  2009.

\bibitem{mag:Dist_RA_games}
D.~Schmidt, C.~Shi, R.~Berry, M.~Honig, and W.~Utschick, ``Distributed resource
  allocation schemes,'' {\em {IEEE} Signal Process. Mag.}, vol.~26, pp.~53
  --63, Sep. 2009.

\bibitem{bk:osborne99a}
M.~J. Osborne and A.~Rubinstein, {\em A Course in Game Theory}.
\newblock {MIT} Press, 1999.

\bibitem{jrnl:cioffi}
W.~Yu, G.~Ginis, and J.~Cioffi, ``Distributed multiuser power control for
  digital subscriber lines,'' {\em {IEEE} J. Sel. Areas Commun.}, vol.~20,
  pp.~1105--1115, Jun. 2002.

\bibitem{jrnl:luo2006analysis}
Z.~Luo and J.~Pang, ``Analysis of iterative waterfilling algorithm for
  multiuser power control in digital subscriber lines,'' {\em EURASIP J. App.
  Sig. Process.}, vol.~6, pp.~1--10, 2006.

\bibitem{jrnl:shum2007convergence}
K.~Shum, K.~Leung, and C.~Sung, ``Convergence of iterative waterfilling
  algorithm for {Gaussian} interference channels,'' {\em {IEEE} J. Sel. Areas
  Commun.}, vol.~25, pp.~1091--1100, Aug. 2007.

\bibitem{jrnl:Etkin2007}
R.~Etkin, A.~Parekh, and D.~Tse, ``Spectrum sharing for unlicensed bands,''
  {\em {IEEE} J. Sel. Areas Commun.}, vol.~25, pp.~517 --528, Apr. 2007.

\bibitem{jrnl:AIWFA}
G.~Scutari, D.~Palomar, and S.~Barbarossa, ``Asynchronous iterative
  water-filling for {Gaussian} frequency-selective interference channels,''
  {\em {IEEE} Trans. Inf. Theory}, vol.~54, pp.~2868--2878, Jul. 2008.

\bibitem{jrnl:KTH}
P.~von Wrycza, M.~R.~B. Shankar, M.~Bengtsson, and B.~Ottersten, ``{Spectrum
  allocation for decentralized transmission strategies: properties of Nash
  equilibria},'' {\em EURASIP J. Adv. Sig. Process.}, vol.~2009, pp.~1--11,
  2009.

\bibitem{jrnl:Blum_TSP03}
S.~Ye and R.~Blum, ``Optimized signaling for {MIMO} interference systems with
  feedback,'' {\em {IEEE} Trans. Signal Process.}, vol.~51, pp.~2839 -- 2848,
  Nov. 2003.

\bibitem{jrnl:Opt1}
G.~Scutari, D.~Palomar, and S.~Barbarossa, ``{Optimal linear precoding
  strategies for wideband noncooperative systems based on game theory -- part
  I: Nash equilibria},'' {\em {IEEE} Trans. Signal Process.}, vol.~56,
  pp.~1230--1249, Mar. 2008.

\bibitem{jrnl:PriceIWFA}
F.~Wang, M.~Krunz, and S.~Cui, ``Price-based spectrum management in cognitive
  radio networks,'' {\em {IEEE} J. Sel. Topics Signal Process.}, vol.~2, pp.~74
  --87, Feb. 2008.

\bibitem{conf:KTH_Gcom09}
P.~von Wrycza, M.~Shankar, M.~Bengtsson, and B.~Ottersten, ``A game theoretic
  approach to multi-user spectrum allocation,'' {\em IEEE Global
  Telecommunications Conference, 2009 (GLOBECOM 2009)}, pp.~1 --6, Nov. 2009.

\bibitem{jrnl:ZQLuo_ComplexityDuality}
Z.-Q. Luo and S.~Zhang, ``Dynamic spectrum management: complexity and
  duality,'' {\em {IEEE} J. Sel. Topics Signal Process.}, vol.~2, pp.~57 --73,
  Feb. 2008.

\bibitem{jrnl:ZQLuo}
S.~Hayashi and Z.-Q. Luo, ``Spectrum management for interference-limited
  multiuser communication systems,'' {\em {IEEE} Trans. Inf. Theory}, vol.~55,
  pp.~1153--1175, Mar. 2009.

\bibitem{conf:Infocom08_DRO}
K.~Yang, Y.~Wu, J.~Huang, X.~Wang, and S.~Verdu, ``Distributed robust
  optimization for communication networks,'' {\em The 27th IEEE International
  Conference on Computer Communications (INFOCOM 2008)}, pp.~1157--1165, Apr.
  2008.

\bibitem{jrnl:robGameTh}
M.~Aghassi and D.~Bertsimas, ``Robust game theory,'' {\em Math. Prog.},
  vol.~107, no.~1, pp.~231--273, 2006.

\bibitem{jrnl:Haykin_RobIWFA}
P.~Setoodeh and S.~Haykin, ``Robust transmit power control for cognitive
  radio,'' {\em Proc. {IEEE}}, vol.~97, pp.~915 --939, May 2009.

\bibitem{jrnl:RIWFA}
R.~Gohary and T.~Willink, ``{Robust IWFA for open-spectrum communications},''
  {\em {IEEE} Trans. Signal Process.}, vol.~57, pp.~4964--4970, Dec. 2009.

\bibitem{conf:Palomar_RobCR}
J.~Wang, G.~Scutari, and D.~Palomar, ``Robust cognitive radio via game
  theory,'' {\em 2010 IEEE International Symposium on Information Theory
  Proceedings (ISIT)}, pp.~2073 --2077, Jun. 2010.

\bibitem{bk:cvxBoyd}
S.~Boyd and L.~Vandenberghe, {\em Convex Optimization}.
\newblock {Cambridge University Press}, March 2004.

\bibitem{jrnl:ellipsoid1}
A.~Ben-Tal and A.~Nemirovski, ``Robust convex optimization,'' {\em Mathematics
  of Operations Research}, vol.~23, pp.~769--805, Nov. 1998.

\bibitem{jrnl:ellipsoid2}
A.~Ben-Tal and A.~Nemirovski, ``Robust solutions of uncertain linear
  programs,'' {\em Operations Research Letters}, vol.~25, pp.~1 -- 13, Aug.
  1999.

\bibitem{bk:ThomasCover}
T.~M. Cover and J.~A. Thomas, {\em Elements of Information Theory}.
\newblock Wiley-Interscience, August 1991.

\bibitem{bk:nisan2007algorithmic}
N.~Nisan, T.~Roughgarden, E.~Tardos, and V.~V. Vazirani, {\em Algorithmic Game
  Theory}.
\newblock New York, NY, USA: Cambridge University Press, 2007.

\bibitem{bk:MatrixAnalysis}
R.~A. Horn and C.~R. Johnson, {\em Matrix Analysis}.
\newblock Cambridge University Press, 1985.

\bibitem{jrnl:Opt2}
G.~Scutari, D.~Palomar, and S.~Barbarossa, ``{Optimal linear precoding
  strategies for wideband noncooperative systems based on game theory -- part
  II: algorithms},'' {\em {IEEE} Trans. Signal Process.}, vol.~56,
  pp.~1250--1267, Mar. 2008.

\bibitem{bk:tsitsiklis1989parallel}
J.~Tsitsiklis and D.~Bertsekas, {\em Parallel and Distributed Computation:
  {Numerical} Methods}.
\newblock Prentice-Hall Englewood Cliffs, NJ, 1989.

\bibitem{Rosen1965}
J.~B. Rosen, ``Existence and uniqueness of equilibrium points for concave
  n-person games,'' {\em Econometrica}, vol.~33, no.~3, pp.~520--534, 1965.

\bibitem{jrnl:JSAC}
G.~Scutari, D.~Palomar, and S.~Barbarossa, ``{Competitive design of multiuser
  MIMO systems based on game theory: A unified view},'' {\em {IEEE} J. Sel.
  Areas Commun.}, vol.~26, pp.~1089--1103, Sep. 2008.

\bibitem{bk:bernstein2009matrix}
D.~Bernstein, {\em Matrix Mathematics: {Theory}, Facts, and Formulas, Second
  Edition}.
\newblock Princeton University Press, 2009.

\bibitem{johnson1970positive}
C.~Johnson, ``Positive definite matrices,'' {\em American Math. Monthly},
  vol.~77, no.~3, pp.~259--264, 1970.

\end{thebibliography}
\end{document}